\documentclass[%
 twocolumn,
 amsmath,amssymb,
prx,
longbibliography
]{revtex4-2}
\usepackage[utf8]{inputenc}
\usepackage{amsmath}
\usepackage{amssymb}
\usepackage{amsbsy}
\usepackage{amsthm}
\usepackage{mathtools}
\usepackage{xcolor}
\usepackage{graphicx}
\usepackage{braket}
\usepackage{acronym}
\usepackage{enumitem}
\usepackage[noend]{algpseudocode}
\usepackage{hyperref}
\hypersetup{colorlinks=true, citecolor=blue, urlcolor=blue, linkcolor=blue}

\newacro{AE}{amplitude estimation}
\newacro{QPE}{quantum phase estimation}
\newacro{ML}{maximum likelihood}
\newacro{CRLB}{Cram\'er--Rao lower bound}

\newtheorem{theorem}{Theorem}[section]
\newtheorem{lemma}[theorem]{Lemma}

\newcommand{\Ntot}{N_{\mathrm{tot}}}
\newcounter{algcnt}
\algrenewcommand\algorithmicrequire{\textbf{Input:}}
\algrenewcommand\algorithmicensure{\textbf{Output:}}

\begin{document}

\title{Quantum amplitude estimation beyond power-of-two schedules}

\author{Farrokh Labib}\email{farrokh@unitary.foundation}\affiliation{Unitary Foundation}

\date{\today}

\begin{abstract}
Non-adaptive quantum amplitude estimation (QAE) fixes its Grover depths in advance, so every circuit can run in parallel, but it has so far needed more queries than the best adaptive methods. We show that most of this gap comes from two conventional choices: subspace-based post-processing and power-of-two depth ladders. We replace the first by the exact maximum-likelihood estimate, one matrix multiplication per batch of estimates, and the second by a geometric ladder with ratio $r \approx 1.45$. The result is a fully parallel, deterministic-schedule estimator with total query complexity $2.8$--$3.1/\varepsilon$ at $95\%$ confidence for target errors from $3.5\times 10^{-3}$ to $1\times 10^{-6}$. This matches the average-case complexity of chebAE, the best benchmarked adaptive method, within statistical uncertainty (with the lower point estimate at every scale tested), beats its maximum-observed complexity by $1.6\times$, and needs a maximum sequential depth of only $0.21/\varepsilon$ against chebAE's $2.9/\varepsilon$. Relative to csAE, the best non-adaptive benchmark, the constants improve by $30$--$35\%$ at $95\%$ and $1.5$--$1.7\times$ at $99\%$ confidence. The optimal ratio has a simple origin. Doubling is the fastest depth growth at which the data can still tell neighboring candidate values apart, so power-of-two ladders sit at the edge of confusion and must buy reliability with extra shots; a slightly denser ladder checks every scale redundantly. An error-probability analysis reproduces the measured failure rates and locates the optimum. The likelihood formulation extends directly to noise-aware estimation, and uniformly scaling the capped ladder covers the depth-limited regime, realizing the optimal trade-off $M \cdot N_{\mathrm{tot}} \approx (0.4\text{--}0.6)/\varepsilon^2$ within $\sim 1.1\times$ of the schedule's Cram\'er--Rao limit.
\end{abstract}

\maketitle

\section{Introduction}\label{sec:intro}

\Ac{AE} \cite{brassard2002quantum} estimates the amplitude $a$ of a marked component of a quantum state, $U\ket{0^l} = \sqrt{1-a^2}\ket{x,0} + a\ket{x',1}$, to additive error $\varepsilon$ using $\mathcal{O}(1/\varepsilon)$ applications of $U$, a quadratic improvement over classical sampling that underlies quantum speedups for Monte Carlo integration \cite{montanaro2015quantum} and financial applications \cite{rebentrost2018quantum, Stamatopoulos2020optionpricingusing, Chakrabarti2021thresholdquantum}. Modern variants avoid \ac{QPE} and instead measure the state after $n$ applications of the Grover iterator $G$ at a set of depths $n$, post-processing the binomial outcome statistics classically \cite{suzuki2020amplitude, aaronson2020quantum, grinko2021iterative, venkateswaran2021quantum, giurgica2022low, Rall2023amplitudeestimation}.

These variants divide into two families. \emph{Adaptive} methods---iterative \ac{AE} \cite{grinko2021iterative}, chebAE \cite{Rall2023amplitudeestimation}, and recent Bayesian variants \cite{ramoa2025bayesian, li2026biqae}---choose each depth from the outcomes so far. Among them chebAE achieves the best benchmarked constants, $\sim 3.1/\varepsilon$ average-case total queries at $95\%$ confidence measured under the conventions of Ref.~\cite{labib2024csae}; the Bayesian variants report improvements over iterative \ac{AE} but have not been benchmarked against chebAE \cite{li2026biqae}. Adaptive methods are inherently sequential: their circuits form one long adaptive chain whose total sequential depth is itself $\Theta(1/\varepsilon)$. \emph{Non-adaptive} methods fix the schedule in advance: maximum-likelihood \ac{AE} (MLAE) \cite{suzuki2020amplitude}, which post-processes the counts by maximizing the likelihood, and csAE \cite{labib2024csae}, which maps the problem to direction-of-arrival estimation and post-processes with ESPRIT on a sparse virtual array. With the schedule fixed, every circuit can be dispatched to parallel processors, and the maximum sequential depth on any one processor can be a small fraction of $1/\varepsilon$. The price has been larger constants: csAE reports $\sim 4.3/\varepsilon$ at $95\%$ confidence, conceding the average case to chebAE.

This work closes most of that gap while keeping full parallelism, using two independent observations.

\emph{The estimator.} For a fixed schedule, the statistically optimal post-processing is the exact \ac{ML} estimate on the raw binomial counts. Because the parameter is one-dimensional, the global \ac{ML} estimate is computable \emph{exactly} (to machine precision) by evaluating the log-likelihood on a grid finer than the spacing of its peaks (the likelihood is multimodal; Sec.~\ref{sec:setting}) and refining the best peak. In vectorized form this is one matrix multiplication per batch of estimates. This removes the sign-recovery heuristic, the virtual-array construction, and the eigendecompositions of csAE entirely, and it is cheaper to compute. On csAE's own published schedule and simulated measurement records (paired trials), exact \ac{ML} improves the $95\%$ constant by $\sim 10\%$ and halves the maximum error. Conversely, the csAE heuristic already operates within $\sim 10\%$ of the information-theoretic ceiling of its schedule: the schedule, not the estimator, is the binding constraint.

\emph{The schedule.} Non-adaptive \ac{AE} schedules have conventionally used power-of-two depth ladders. We show that this choice is significantly suboptimal, for a reason that appears not to have been noticed. Doubling is the \emph{maximum} depth growth per rung at which adjacent likelihood basins (the peaks of the multimodal likelihood, one per candidate value the deepest depth cannot distinguish; Sec.~\ref{sec:setting}) remain statistically distinguishable. A power-of-two ladder therefore makes every branch decision marginal, and the resulting basin flips must be suppressed with large shot counts at shallow depths, where shots contribute little Fisher information. A geometric ladder with ratio $r \approx 1.45$ instead resolves each scale redundantly; catastrophic failures, estimates landing many basins from the truth, become rare ($\lesssim 10^{-4}$ at the optimum) without extra shallow shots, and the freed budget migrates to depth. A Chernoff-bound analysis of this aliasing (fringe realignment; Sec.~\ref{sec:ladder}) makes the mechanism quantitative and reproduces the measured failure orderings. The optimum is an interior one: below $r \approx 1.3$ the redundancy itself becomes the dominant cost, through a Fisher-efficiency factor $(r+1)/(r-1)$.

The resulting method is simple: a geometric ladder, a linear shot profile with a single shot at the deepest level (the ``canonical profile'' of Sec.~\ref{sec:ladder}), and one matrix multiplication of classical post-processing. In the standard shot-noise model at amplitudes $a \in (0.1, 0.9)$, it achieves total query constants of $2.8$--$3.1/\varepsilon$ at $95\%$ confidence across target errors from $3.5\times 10^{-3}$ to $1\times 10^{-6}$, with maximum sequential depth $\approx 0.21/\varepsilon$ (Sec.~\ref{sec:numerics}). In a matched head-to-head using the chebAE implementation of Ref.~\cite{Rall2023amplitudeestimation}, this equals or beats chebAE's average-case constants at every scale tested while beating its maximum-observed constants by $1.5$--$1.7\times$; chebAE's sequential depth is $\sim 13\times$ larger. At $99\%$ confidence the constant is $5.1$--$5.8$, matching chebAE's average. A Ziv--Zakai bound (Sec.~\ref{sec:zzb}) shows that the exact \ac{ML} estimate is within $4$--$6\%$ of the best any estimator can achieve on this schedule at $68\%$ and $95\%$ confidence, so the remaining room is in the schedule. Finally, the likelihood formulation makes noise-aware estimation straightforward: a known depolarizing rate enters the model probabilities in one line (Sec.~\ref{sec:noise}).

\section{Setting, metrics, and information}\label{sec:setting}

We use the standard \ac{AE} setting \cite{brassard2002quantum, labib2024csae}. Given $U\ket{0^l} = \cos\theta\ket{x,0} + \sin\theta\ket{x',1}$ with $a = \sin\theta$, $\theta \in (0, \pi/2)$, the Grover iterator $G$ yields, after $n$ applications, a state whose final-qubit measurement returns $0$ with probability
\begin{equation}\label{eq:model}
    p_n(\theta) = \cos^2\!\big((2n+1)\theta\big).
\end{equation}
A \emph{schedule} is a set of depths $D = \{n_1 < \dots < n_L\}$ (the \emph{rungs} of the ladder; we always include $n_1 = 0$, and the deepest depth is written $n_L$ or, in tables and figures, $n_{\max}$) with shot counts $\{N_j\}$; the data are independent binomial counts $k_j \sim \mathrm{Bin}(N_j, p_{n_j}(\theta))$. Following Ref.~\cite{labib2024csae}, the total query count is $\Ntot = \sum_j N_j n_j + N_1$: each application of $G$ counts as one query, and the depth-$0$ circuits, which contain no $G$, count their state preparation $U$ instead, which gives the $N_1$ term. We quote performance as the constant $C_\delta = \varepsilon_\delta \cdot \Ntot$, where $\varepsilon_\delta$ is the $\delta$th percentile of $|\hat a - a|$ over amplitudes drawn uniformly from $(0.1, 0.9)$. The \emph{parallel} constant is $\varepsilon_\delta \cdot n_L$: since the schedule is fixed in advance, all circuits can be distributed over processors, and $n_L$ is the unavoidable sequential depth of the deepest one.

The information content of a schedule has a simple closed form.

\begin{lemma}\label{lem:fisher}
The Fisher information about $\theta$ carried by $N$ shots at depth $n$ is
\begin{equation}
    I_n \;=\; N\,\frac{(\partial_\theta p_n)^2}{p_n(1-p_n)} \;=\; 4N(2n+1)^2,
\end{equation}
independent of $\theta$.
\end{lemma}

\begin{proof}
With $\phi = (2n+1)\theta$: $\partial_\theta p_n = -(2n+1)\sin 2\phi$ and $p_n(1-p_n) = \cos^2\!\phi\,\sin^2\!\phi = \tfrac14 \sin^2 2\phi$; the $\sin^2 2\phi$ factors cancel.
\end{proof}

Two consequences shape everything that follows. First, schedule design decouples from the unknown value: there are no ``bad fringes,'' and no need to steer measurements onto slopes, adaptively or otherwise. Second, the \ac{CRLB} $\sigma_\theta \ge (\sum_j 4N_j(2n_j+1)^2)^{-1/2}$ implies that information per query grows linearly with depth ($I_n/(Nn) \approx 16n$), so an unconstrained optimizer would push all shots to the deepest level. What prevents this is not local information but \emph{global identifiability}: the likelihood of Eq.~\eqref{eq:model} is multimodal, with basins of width $\pi/(2(2n_L+1))$ set by the deepest level, and shallow levels exist solely to reject the rival basins. Quantifying that rejection is the subject of Sec.~\ref{sec:ladder}.

Converting the bound into the amplitude metric needs care: $|\hat a - a| \simeq \cos\theta\,|\hat\theta - \theta|$ is a scale \emph{mixture} of half-normals, as $\theta$ is itself drawn from the prior, so the percentile must be taken after the mixing. It is $\kappa_\delta \sigma_\theta$ with $\mathbb{E}_a[2\Phi(\kappa_\delta/\cos\theta) - 1] = \delta$, giving $\kappa_{68} = 0.802$, $\kappa_{95} = 1.669$, $\kappa_{99} = 2.258$ at $a \sim \mathcal{U}(0.1, 0.9)$; the natural-looking $z_\delta\langle\cos\theta\rangle$, with $z_\delta$ the standard-normal quantile, averages the scale first and understates the floor by $3.7\%$ at $\delta = 95$.

The \ac{CRLB} is a local bound, and it cannot be attained once rival basins survive with non-negligible probability. For our flagship schedule (the $r = 1.45$ ladder plus one extra rung near the cap, defined in Sec.~\ref{sec:numerics}) it gives $C_{95} \approx 2.4$, and the measured constants of the unscaled schedule sit $\approx 1.23\times$ above it. Two effects account for that excess. First, rival basins occasionally survive, and the bound ignores them. Second, at single-digit shot counts the likelihood within a basin is noticeably skewed, and the \ac{ML} rule of reporting its peak loses a little precision relative to reporting its center of mass. Heavy tails play no role: at the $95$th percentile the within-basin error distribution is Gaussian to within a few percent, so the excess is extra variance, not a change of shape. When all shot counts are multiplied by a common factor at a fixed depth cap (the uniform scaling of Sec.~\ref{sec:depth}), the ratio falls to $1.02$--$1.11$.

\section{Exact global maximum likelihood}\label{sec:ml}

For a fixed schedule the log-likelihood is
\begin{equation}\label{eq:loglik}
    \ell(\theta) = \sum_{j} \Big[ k_j \ln p_{n_j}(\theta) + (N_j - k_j)\ln\big(1 - p_{n_j}(\theta)\big) \Big],
\end{equation}
and we define $\hat\theta$ as its global maximizer on $(0, \pi/2)$, with $\hat a = \sin\hat\theta$. The likelihood $e^{\ell}$ is an even, $\pi$-periodic trigonometric polynomial in $\theta$. Its modes have Gaussian standard deviation $\sigma_\theta = (\sum_j 4N_j(2n_j+1)^2)^{-1/2}$, which for the ladders used here is between a sixth and an eighth of the deepest-level basin width $\pi/(2(2n_L+1))$; this ratio depends on the shape of the ladder alone. A uniform grid at one eighth of a basin therefore samples every mode near its peak. A two-stage local refinement (a zoom grid across the best coarse step, followed by quadratic interpolation) then recovers $\hat\theta$ to far below the statistical error. The cost is a single $(\text{trials} \times \text{grid})$ matrix product, $\mathcal{O}(L\, n_L)$ arithmetic: $\approx 5\times 10^{8}$ floating-point operations and $\approx 10$~ms per estimate for the deepest schedule used here ($n_{\max} \approx 2.1\times 10^5$). That is far below the classical cost of the ESPRIT pipeline it replaces, and negligible against the quantum cost. Appendix~\ref{app:ml} gives the mode-width identity, the grid-spacing rule that keeps the search exact when all shot counts are scaled up (Sec.~\ref{sec:depth}), and the cost accounting. It also checks the search numerically: the rare trials on which a $32\times$ finer search disagrees are statistical near-ties, and the error percentiles are unchanged.

This estimator is MLAE \cite{suzuki2020amplitude} implemented with an exhaustive global search. We emphasize this because MLAE is often reported as substantially inferior to later methods, at $\approx 10/\varepsilon$ against $\approx 4.3/\varepsilon$ for chebAE in the benchmark of Ref.~\cite{Rall2023amplitudeestimation}. Those benchmarks run MLAE with $100$ shots on every rung of a power-of-two ladder, a schedule whose own \ac{CRLB} constant is $C_{95} \approx 7$; the reported gap is a property of that shot-padded schedule, not of likelihood post-processing. On identical schedules the exact global maximizer is at least as good as any other post-processing, as the following comparison shows. In particular, we ran exact \ac{ML} and the full csAE pipeline on identical measurement records for the published csAE schedule (powers of two to depth $128$, $2080$ queries, $500$ paired simulated trials reproducing the csAE repository data trial-for-trial). Exact \ac{ML} gives $C_{95} = 3.92$ against $4.37$, and maximum error $1.6\times 10^{-2}$ against $3.4\times 10^{-2}$. This has two consequences. The sign-recovery and ESPRIT steps of Ref.~\cite{labib2024csae} can be dropped, for a $\sim 10\%$ gain. Conversely, those steps were already extracting $\sim 90\%$ of the information their schedule offers, so larger improvements must come from the schedule.

\begin{figure}[t]
\refstepcounter{algcnt}\label{alg:method}
\vspace{2pt}\hrule\vspace{3pt}
\noindent\textbf{Algorithm \thealgcnt.} Non-adaptive amplitude estimation.
\vspace{3pt}\hrule\vspace{4pt}
\begin{algorithmic}[1]
\Require maximum Grover depth $n_{\max}$; ladder ratio $r = 1.45$
\Statex \hspace{-1.4em}\emph{Schedule} --- fixed in advance, uses no data
\State $D \gets \{0\} \cup \{\,\mathrm{round}(r^{\,j}) \le n_{\max}\,\}$, distinct
\State $D \gets D \cup \{\mathrm{round}(\max(D)/1.3)\}$ \Comment{breaks the cap's mirror symmetry, Eq.~\eqref{eq:extrema}}
\State sort $D = \{n_1 < \dots < n_L\}$; \; $N_j \gets L-j+1$
\Statex \hspace{-1.4em}\emph{Data} --- all $L$ circuits may run in parallel
\State \textbf{for each} $j$: measure $G^{n_j}U\ket{0^l}$ in the $Z$ basis $N_j$ times, record $k_j$
\Statex \hspace{-1.4em}\emph{Estimate}
\State $K \gets (k_1, \dots, k_L)$;\; $N \gets (N_1, \dots, N_L)$ \Comment{row vectors}
\State $w \gets \pi/[2(2n_L{+}1)]$;\; $\sigma_\theta \gets \big(\textstyle\sum_j 4N_j(2n_j{+}1)^2\big)^{-1/2}$
\State $\Delta \gets \min(w/8,\ 4\sigma_\theta)$;\; grid $\{\theta_g\} \subset (0,\pi/2)$, spacing $\Delta$
\State $P_{gj} \gets p_{n_j}(\theta_g)$
\State $\mathrm{NLL} \gets -\big[K\log P^{\!\top} + (N{-}K)\log(1{-}P)^{\!\top}\big]$ \Comment{one matrix product}
\State $g^\star \gets \arg\min_g \mathrm{NLL}_g$
\State evaluate the log-likelihood $\ell$ [Eq.~\eqref{eq:loglik}] at $25$ points in $\theta_{g^\star}\pm\Delta/2$; parabola through the best three $\to\hat\theta$
\Ensure $\hat a = \sin\hat\theta$
\end{algorithmic}
\vspace{3pt}\hrule\vspace{2pt}
\end{figure}

A practical remark: unlike subspace methods, the likelihood \eqref{eq:loglik} needs no signs, no uniform (virtual) array, no covariance model, and no Gaussianity---it is the exact finite-shot model. It also composes with hardware-noise models by replacing $p_{n}(\theta)$ with the noisy response, which we exploit in Sec.~\ref{sec:noise}.

\section{Ladder design and the aliasing mechanism}\label{sec:ladder}

\begin{figure}[t]
\centering
\includegraphics[width=0.98\columnwidth]{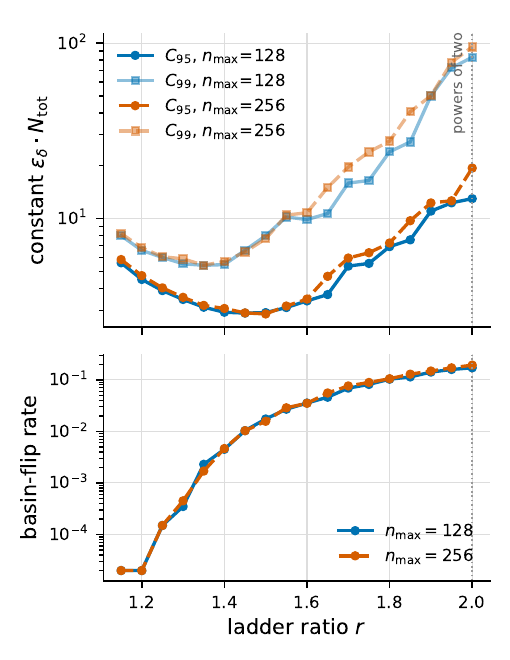}
\caption{\label{fig:ratio}Total-query constants (top) and basin-flip rate (bottom) versus the ladder ratio $r$, for geometric ladders $n_j = \mathrm{round}(r^j)$ capped at $n_{\max} \in \{128, 256\}$ with the canonical shot profile (one shot at the deepest level, one more per shallower level), $2\times 10^4$ trials per point. Both scales collapse onto the same curves. $C_{95}$ is minimized at $r \approx 1.45$--$1.5$ and $C_{99}$ at $r \approx 1.35$; at $r = 2$ (powers of two, dotted line) one trial in six lands in a wrong basin under this profile. Bootstrap $95\%$ confidence intervals are smaller than the markers for $C_{95}$ and up to $\sim 15\%$ for $C_{99}$ at large $r$; flip rates below $10^{-3}$ rest on single-digit event counts and indicate $\approx 0$ rather than precise values.}
\end{figure}

Fix a target maximum depth $n_L$ and consider geometric ladders $n_j \approx r^j$. Figure~\ref{fig:ratio} shows the measured constants and the \emph{basin-flip rate}---the fraction of trials whose estimate is more than one deepest-level basin width from the truth---as functions of $r$, under a deliberately untuned shot profile: one shot at the deepest level and one more per shallower level. We call this the \emph{canonical profile} and use it throughout. Three regimes are visible. For $r \gtrsim 1.8$ the flip rate reaches tens of percent and the constants are dominated by catastrophic errors, which we define as flips by more than ten basins. Near $r = 1.45$ the flip rate is $\sim 1\%$, essentially all of it benign single-basin slips of size comparable to the resolution itself, and $C_{95}$ attains its minimum $\approx 2.9$. Below $r \approx 1.3$ flips are absent but the constants rise again: redundancy is no longer free.

Both sides of the optimum can be explained quantitatively.

\emph{Large $r$: aliasing failures.} The probability that the likelihood prefers a rival angle $\theta + d$ is bounded by the Chernoff/Bhattacharyya bound
\begin{equation}\label{eq:chernoff}
    \Pr[\ell(\theta + d) \ge \ell(\theta)] \;\le\; e^{-E(\theta, d)},\quad
    E = -\sum_j N_j \ln \mathrm{BC}_j,
\end{equation}
with $\mathrm{BC}_j = \sqrt{p_j p_j'} + \sqrt{(1-p_j)(1-p_j')}$, $p_j = p_{n_j}(\theta)$, $p_j' = p_{n_j}(\theta+d)$. A depth $n$ discriminates offsets $d \gtrsim 1/(2n)$ at $\mathcal{O}(1)$ exponent per shot and smaller offsets only quadratically. A rival at offset $d$ is therefore rejected by the ladder levels above $\sim 1/d$, \emph{except} where the fringe patterns of those levels realign simultaneously, i.e., $(2n_j+1)d \approx 0 \pmod{\pi}$ for all deep levels at once. For $r = 2$ such realignments are systematic: they leave notches of near-zero exponent at rival offsets tens of basins away (Fig.~\ref{fig:alias}). These are the catastrophic failures. Suppressing them at $r=2$ requires the large shallow-shot budgets typical of published power-of-two schedules, such as the $72$-shot depth-$0$ level of Ref.~\cite{labib2024csae}. At $r = 1.45$ the incommensurate fringe periods keep the exponent floor high at the same total cost. The bound \eqref{eq:chernoff}, summed over rival basins, is loose in absolute terms (as union--Chernoff bounds are) but reproduces the measured \emph{ordering} of catastrophic-failure rates: a drop of two to three orders of magnitude from $r = 2$ to $r = 1.45$ (Fig.~\ref{fig:alias}).

\emph{Small $r$: redundancy cost.} Index the rungs from the deepest, so $n_k \approx n_L r^{-k}$ and the canonical profile assigns $N_k = k+1$ shots. With $\sum_k (k+1)x^k = (1-x)^{-2}$ this gives $\Ntot \approx n_L\, r^2/(r-1)^2$ and $\sum_k I_{n_k} \approx 16 n_L^2\, r^4/(r^2-1)^2$, so the \ac{CRLB}-limited constant is
\begin{equation}\label{eq:fisherflank}
    C_\delta^{\mathrm{CRLB}} \;=\; \kappa_\delta\,\frac{\Ntot}{\sqrt{\sum_j I_{n_j}}} \;\approx\; \frac{\kappa_\delta}{4}\cdot\frac{r+1}{r-1},
\end{equation}
which reproduces the exact \ac{CRLB} constant of the integer-rounded ladders to within $1$--$12\%$ at the $n_{\max} \approx 128$ scale of Fig.~\ref{fig:ratio} (worst toward $r = 1.2$, where integer rounding distorts the short ladder most) and to $\approx 1\%$ across $r \in [1.2, 2]$ once $n_{\max} \gtrsim 2\times 10^3$. It diverges as $r \to 1$: each extra rung adds queries faster than it adds information. The exponent depends on the shot profile, not on the ladder alone: with a constant number of shots per rung, the factor $(r+1)/(r-1)$ in Eq.~\eqref{eq:fisherflank} would be replaced by its square root. The linear ramp therefore makes the redundancy penalty steeper than a flat allocation would. The interior optimum is the product of a failure term that falls with $r$ and the factor \eqref{eq:fisherflank} that rises as $r \to 1$. Empirically the two balance at $r \approx 1.45$ for the $95\%$ constant and at $r \approx 1.35$ for the $99\%$ constant; the tails prefer more redundancy.

We also tested arithmetic (rather than density) mechanisms: ladders with $(2n_j+1)$ pairwise-coprime primes, and randomly jittered power-of-two ladders, perform no better than powers of two at matched cost. Density is what matters.

\emph{One rung off the geometric law.} Everything above concerns the pure geometric ladder, which is the object that admits the closed forms \eqref{eq:chernoff} and \eqref{eq:fisherflank}. The schedule we actually report in Sec.~\ref{sec:numerics} departs from it in exactly one place: a single additional rung at $\mathrm{round}(n_{\max}/1.3)$, sitting in the gap below the deepest rung. Its purpose is not covered by either side of the analysis above, which describes how well the ladder rejects \emph{rival basins}. The extra rung fixes a different failure. The deepest level's response is exactly symmetric about each of its fringe extrema (Eq.~\eqref{eq:extrema}, analyzed in Sec.~\ref{sec:depth}). Its data therefore cannot tell $\theta$ from its mirror image about the nearest extremum. When the rest of the schedule fails to break that tie, the estimate lands a mirror image away, and such errors make up the far tail of the error distribution. The extra rung places a second fringe pattern, nearly as deep and with an incommensurate period, exactly where the deepest one is blind. Section~\ref{subsec:heisenberg} quantifies what it buys. We keep the plain ladder and the flagship separate for two reasons. First, the analysis of this section is about the plain geometric ladder. Second, the rung is only needed for the unscaled ladder: under the uniform scaling of Sec.~\ref{sec:depth}, multiplying all shot counts by $s$ drives every failure probability, mirror ties included, down exponentially in $s$, so the capped ladder needs no such repair beyond the cap rung itself.

\begin{figure}[t]
\centering
\includegraphics[width=0.98\columnwidth]{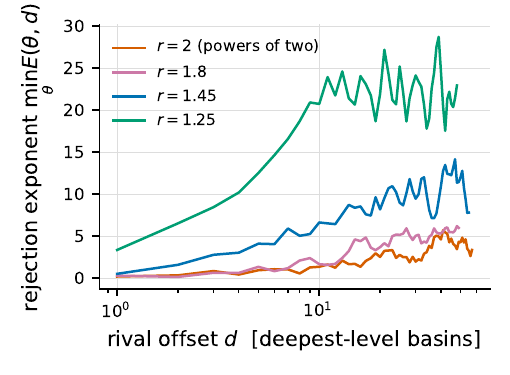}
\caption{\label{fig:alias}Worst-case (over $\theta$) rejection exponent $E$ of Eq.~\eqref{eq:chernoff} against rival offsets $d$, in units of deepest-level basins, for four ladders at $n_{\max} \approx 128$ with the canonical shot profile. Powers of two ($r=2$) and $r=1.8$ leave near-zero exponents---unrejected rivals---out to tens of basins; $r = 1.45$ and $1.25$ reject all far rivals strongly. These profiles reproduce the measured ordering of catastrophic-failure rates: for far rivals ($>10$ basins) at matched profiles, $2.7\%$ ($r{=}2$), $0.7\%$ ($r{=}1.8$), $0.005\%$ ($r{=}1.45$), $\sim 0$ ($r{=}1.25$); the last two rest on one and zero events in $2\times 10^4$ trials.}
\end{figure}

\section{Schedules for depth-limited operation}\label{sec:depth}

Early fault-tolerant devices may only support Grover depths $M$ far below the $\Theta(1/\varepsilon)$ needed at the Heisenberg limit. In that case the best possible trade-off is $M \cdot \Ntot \approx \varepsilon^{-2}$ up to polylogarithmic factors \cite{giurgica2022low, huang2026eigengap, erle2026anydepth}: halving the depth doubles the total query count. Existing proofs establish this scaling but not its constants. Our ladders extend directly to this regime. We quote performance as the trade-off constant
\begin{equation}
    \widetilde C_\delta \;:=\; \varepsilon_\delta^2\, M\, \Ntot;
\end{equation}
the \ac{CRLB} floor for our ladder shape is $\widetilde C_{95} \approx 0.4$, and the value for the unscaled ladder ($s = 1$) follows from Sec.~\ref{sec:numerics}: $\widetilde C_{95} = C_{95}\cdot(\varepsilon_{95} n_L) \approx 2.9 \times 0.21 \approx 0.61$ for the flagship (the capped ladders of Table~\ref{tab:depth} start at $0.6$--$0.8$).

Our schedule for this regime is \emph{uniform scaling}: take the $r = 1.45$ ladder truncated at $M$, with $M$ itself appended as the deepest rung, and the canonical shot profile (cost $B_0$), and multiply every shot count by the same factor, $N_j = s\, N_j^{(0)}$ with $s = B/B_0$, until the budget $B$ is spent. The rest of this section explains why this rule, rather than a top-loaded one, is the right choice.

At first sight it looks wasteful. By Lemma~\ref{lem:fisher}, shots at the cap carry the most information per query, so one would expect to spend the extra budget at depth $M$ and keep the shallow levels at their minimum. But information is not the same as \emph{identifiability}. Data taken at one depth have exact symmetries. Since $\partial_\theta p_M = -(2M{+}1)\sin\big(2(2M{+}1)\theta\big)$, the response is stationary at the fringe extrema
\begin{equation}\label{eq:extrema}
    \theta_{\mathrm{ext}} = \frac{k\pi}{2(2M+1)}, \qquad k \in \mathbb{Z},
\end{equation}
where $p_M$ equals $1$ or $0$ and which are spaced by exactly one deepest-level basin width. About each of them $p_M$ is exactly even, $p_M(\theta_{\mathrm{ext}} + u) = p_M(\theta_{\mathrm{ext}} - u)$ for all $u$, so depth-$M$ shots can never tell $\theta$ from its mirror image about the nearest such point---and small sets of depths retain similar shared blind spots. Whenever the truth falls near such a point, the tie must be broken by the rest of the schedule. Piling shots onto one depth therefore sharpens the likelihood without removing its ambiguities, and the consequences are graded. Call a large block of shots placed at a single depth an \emph{anchor}. For a single anchor at the cap the failure is total. The mirror image sits up to a full basin away, so $\varepsilon_{95}$ stalls at $\varepsilon_{95} M \approx 0.14$--$0.28$ no matter how large the anchor. The tails even \emph{grow} with it: the sharp cap likelihood locks confidently onto a mirror basin, and $\widetilde C_{99}$ diverges. Splitting the anchor over two nearby depths, $\{0.8M, M\}$ say, breaks the exact mirror, and does so well: in our experiments the $95$th percentile then tracks the total \ac{CRLB} ($\varepsilon_{95}/\varepsilon_{95}^{\mathrm{CRLB}} \approx 1.0$) across a $256\times$ anchor range, reaching $\widetilde C_{95} \approx 0.19$. But the blind spots are thinned, not removed. Wherever the two anchor fringes realign (a discrete set of angles fixed by the arithmetic of $2n+1$ at the two depths), the tie is still decided by the unscaled remainder of the schedule, which fails with a small, budget-independent probability. The measure of these blind spots shrinks only like the anchor's likelihood width, that is, \emph{polynomially} in the budget, so a thin tail of multi-basin errors survives; Sec.~\ref{subsec:depthnum} quantifies it. Suppressing these residual ambiguities one by one, with guarantees, is precisely where the polylogarithmic factors of the rigorous constructions come from.

Uniform scaling removes all such failures at once. Scaling all shots by $s$ multiplies every log-likelihood ratio by $s$, so every rejection exponent in Eq.~\eqref{eq:chernoff} grows as $E \to sE$. \emph{Every} failure probability (wrong basins, mirrors, blind spots) therefore dies off exponentially in $s$, all at once, whereas the anchored designs above thin their blind spots only polynomially. The error then follows the total \ac{CRLB} at a fixed small multiple at every scale. Two further properties follow. First, the \ac{CRLB} floor of $\widetilde C_\delta$ is independent of $s$: precision and budget both scale, and their combination cancels; the measured constant converges to that floor for $s \gtrsim 4$, so one measured constant covers all budgets. Second, the rule inherits the design work of Sec.~\ref{sec:ladder} unchanged: relative shot allocations are preserved, so the balance between redundancy and Fisher efficiency that fixed $r = 1.45$ holds at any $s$.

Uniform scaling does have a cost. The ladder shape fixes the ratio of Fisher information to query cost, giving $\widetilde C_{95}^{\mathrm{CRLB}} \approx 0.37$--$0.47$. Concentrating all information at the cap would in principle allow $\approx 0.17$, and the anchored pair above shows that this level is approachable at the $95$th percentile. What the $\sim 2\times$ premium of uniform scaling buys is exponential, budget-uniform suppression of the catastrophic tail, with no placement arithmetic to keep track of. Section~\ref{subsec:depthnum} reports the resulting constants, within $\sim 1.1\times$ of the \ac{CRLB} with no tuning and no validity window, and the tails of both families.

\section{Numerical results}\label{sec:numerics}

\begin{figure}[t]
\centering
\includegraphics[width=0.98\columnwidth]{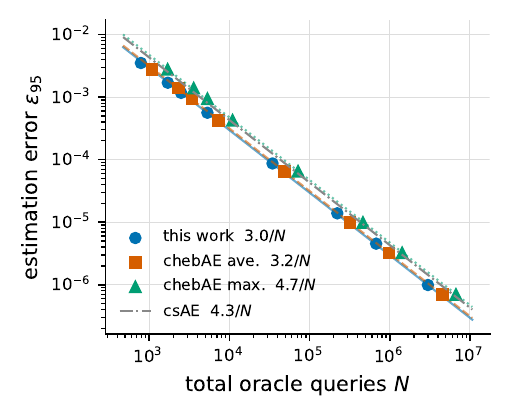}
\caption{\label{fig:scaling}Estimation error at $95\%$ confidence versus total oracle queries, down to $\varepsilon_{95} \approx 10^{-6}$. Circles (blue): this work ($r=1.45$ ladders, canonical shots, $10^4$--$5\times 10^4$ trials per point). Squares/triangles: chebAE average and maximum-observed queries at matched achieved error ($2000$ trials per point, using the implementation of Ref.~\cite{Rall2023amplitudeestimation}). Gray line: the published csAE constant \cite{labib2024csae}.}
\end{figure}

\emph{Protocol.} All simulations draw amplitudes uniformly from $(0.1, 0.9)$ and sample exact binomial counts from Eq.~\eqref{eq:model} (the same model as Refs.~\cite{suzuki2020amplitude, grinko2021iterative, labib2024csae}). Schedules were selected on search seeds and all reported numbers come from disjoint holdout seeds with $2\times 10^4$--$10^6$ trials (stated per table). Uncertainties on all reported constants are nonparametric bootstrap $95\%$ confidence intervals, obtained by resampling the trials with replacement and recomputing the constant. They quantify the finite-trial uncertainty of our estimates and are distinct from the confidence level $\delta$ of the estimated quantity itself. For example, Table~\ref{tab:scale} reports the $99$th-percentile constant $C_{99}$ with a $95\%$ confidence interval around it. The code and data for reproducing every table and figure in this paper, with instructions, can be found at~\href{https://github.com/unitaryfoundation/csAE}{https://github.com/unitaryfoundation/csAE}.

\subsection{Heisenberg-limit schedules}\label{subsec:heisenberg}

\emph{Flagship schedule.} Table~\ref{tab:scale} reports the $r = 1.45$ ladder with the canonical profile and one extra rung at $\mathrm{round}(n_{\max}/1.3)$, across a $28\times$ range of maximum depths. The $95\%$ constant is flat at $C_{95} = 2.78$--$3.02$, the $99\%$ constant is $5.1$--$5.8$, and the parallel constant is $\approx 0.21$. No shot-count growth with scale is required: sweeping the deepest-level shot count at every scale shows a single shot remains optimal for $C_{95}$ throughout. Extending the same rule to deeper schedules holds the constants flat over three and a half decades of target error, $C_{95} = 2.96$--$3.07$ down to $\varepsilon_{95} = 9.9\times 10^{-7}$ ($n_{\max} \approx 2.1\times 10^{5}$, $3.0\times 10^{6}$ queries; Fig.~\ref{fig:scaling}).

\emph{What the extra rung buys.} Section~\ref{sec:ladder} analyzes the plain geometric ladder; here we quantify what the one departure from it contributes. Against that plain ladder at matched scale ($5\times 10^4$ paired trials), the rung lowers the $99\%$ constant by $16$--$20\%$ ($6.25 \to 5.27$ at $n_{\max} = 125$; $7.05 \to 5.66$ at $n_{\max} = 382$) and the parallel constant by $\approx 27\%$ ($0.286 \to 0.210$) at no cost in $C_{95}$ ($2.84 \to 2.88$), paying $\approx 7\%$ of $C_{68}$ ($1.19 \to 1.27$). The gain is not simply due to the extra shots: spending the same $\Ntot$ on uniformly scaling the plain ladder instead gives $C_{95} = 3.14$ and $C_{99} = 5.86$, worse on both. Nor is the placement finely tuned: any divisor in $[1.1, 2.5]$ gives $C_{99} = 5.1$--$5.7$ against the plain ladder's $6.25$, though rungs nearer the cap favor $C_{95}$ and the parallel constant and rungs further down favor $C_{99}$.

\begin{table*}[t]
\caption{\label{tab:scale}Certified performance of the flagship schedule---the $r=1.45$ geometric ladder plus one extra rung at $\mathrm{round}(n_{\max}/1.3)$, canonical shots---across a $28\times$ range of maximum depth. The parallel constant $\varepsilon_{95}n_{\max}$ is $0.211$--$0.214$ throughout. Each column is one schedule; $10^6$ holdout trials each. $C_\delta = \varepsilon_\delta \Ntot$ is the total-query constant at confidence $\delta$. Brackets are bootstrap $95\%$ confidence intervals (protocol in Sec.~\ref{sec:numerics}). At $10^6$ trials they are narrower than the last digit shown for $C_{68}$ and $C_{95}$ (at most $\pm 0.003$ and $\pm 0.009$, respectively), so only those on $C_{99}$ are given. Statistical precision is no longer the limiting uncertainty at this trial count: widening the amplitude prior to $\mathcal{U}(0.01, 0.99)$ moves $C_{95}$ by up to $1.6\%$, roughly five times more.}
\begin{ruledtabular}
\begin{tabular}{lcccccc}
$n_{\max}$ & $60$ & $125$ & $182$ & $382$ & $803$ & $1688$ \\
\hline
$\Ntot$ & $789$ & $1708$ & $2502$ & $5325$ & $11272$ & $23781$ \\
$\varepsilon_{95}$ & $3.5\times 10^{-3}$ & $1.7\times 10^{-3}$ & $1.2\times 10^{-3}$ & $5.6\times 10^{-4}$ & $2.7\times 10^{-4}$ & $1.3\times 10^{-4}$ \\
$C_{68}$ & $1.222$ & $1.270$ & $1.284$ & $1.307$ & $1.317$ & $1.321$ \\
$C_{95}$ & $2.776$ & $2.886$ & $2.918$ & $2.972$ & $2.997$ & $3.017$ \\
$C_{99}$ & $5.11$ & $5.32$ & $5.44$ & $5.66$ & $5.77$ & $5.80$ \\
         & $[5.07, 5.15]$ & $[5.28, 5.37]$ & $[5.40, 5.49]$ & $[5.61, 5.71]$ & $[5.71, 5.82]$ & $[5.74, 5.85]$ \\
\end{tabular}
\end{ruledtabular}
\end{table*}

\emph{Head-to-head with chebAE.} We ran the chebAE implementation of Ref.~\cite{Rall2023amplitudeestimation} under the same protocol at each of eight scales, targeting our achieved $\varepsilon_{95}$ with its failure parameter set to $0.05$ ($2000$ trials per scale). Its constants are its own achieved $95$th-percentile error times its average query count, or times its maximum-observed query count, following the conventions of Ref.~\cite{labib2024csae}. The results are shown in Fig.~\ref{fig:scaling} and Table~\ref{tab:h2h}. The chebAE average-case constants are $3.00$--$3.23$; ours are $2.76$--$3.07$, with a lower point estimate at all eight scales. The chebAE maximum-observed constants are $4.6$--$5.0$; our query count is deterministic. Its sequential-depth constant (its adaptive rounds cannot be parallelized) is $2.8$--$3.0$, against $0.21$ here. We repeated the comparison at the other two confidence levels, running chebAE with failure parameter $0.32$ and $0.01$ against our achieved $\varepsilon_{68}$ and $\varepsilon_{99}$. At $68\%$, chebAE averages $1.41\ [1.32, 1.47]$ and $1.27\ [1.19, 1.35]$ at the $n_{\max} = 125$ and $182$ scales, against $1.26\ [1.25, 1.28]$ and $1.29\ [1.28, 1.30]$ here: a win at the first scale and a tie at the second. At $99\%$, chebAE averages $5.47\ [4.74, 7.33]$ against $5.39\ [5.22, 5.57]$ here, a tie. In summary, a fully parallel deterministic schedule now matches or beats the average-case cost of the best adaptive method at every confidence level tested, while needing $13\times$ less sequential depth.

\subsection{Depth-limited schedules}\label{subsec:depthnum}

Table~\ref{tab:depth} reports the uniform-scaling rule of Sec.~\ref{sec:depth} at caps $M \in \{64, 256\}$ over budgets spanning a $64\times$ range: for $s \ge 4$, $\widetilde C_{95} = 0.39$--$0.58$ and $\widetilde C_{99} = 0.73$--$1.35$, with $\varepsilon_{95}$ within $1.02$--$1.11\times$ of the total \ac{CRLB} and no sign of a plateau; that is,
\begin{equation}
    \Ntot \;\approx\; \frac{0.39\text{--}0.58}{M\,\varepsilon_{95}^2}
\end{equation}
uniformly---the trade-off frontier of Refs.~\cite{giurgica2022low, huang2026eigengap, erle2026anydepth} realized with explicit constants by a one-line schedule rule. Figure~\ref{fig:fan} displays the resulting family of trade-offs at caps $M \in \{16, 64, 256, 1024\}$. Each fixed cap follows its own $\varepsilon \propto N^{-1/2}$ line and leaves the Heisenberg envelope near its endpoint budget. A depth-growth policy $M \propto N^{\beta}$ is a path across this fan (caption of Fig.~\ref{fig:fan}); classical sampling ($\beta = 0$) and the Heisenberg limit ($\beta = 1$) are its two edges. The anchored alternatives of Sec.~\ref{sec:depth} behave as described there. A single anchor at the cap stalls at $\varepsilon_{95}M \approx 0.14$--$0.28$ (measured up to anchors of $10^6$ shots), its tails growing with the anchor ($\widetilde C_{99}$ in the hundreds and rising). Pair anchors, and anchors spread over a narrow band of depths below the cap, beat uniform scaling at the $95$th percentile, with $\widetilde C_{95} \approx 0.19$--$0.23$ at $\varepsilon_{95}/\varepsilon_{95}^{\mathrm{CRLB}} \approx 1.0$. But they carry the polynomially suppressed catastrophic tail of Sec.~\ref{sec:depth}. At $10^{5}$--$10^{6}$ trials we measure wrong-basin rates of $10^{-5}$--$10^{-4}$ with errors hundreds of times $\varepsilon_{95}$, and $\widetilde C_{99.9}$ between $0.6$ and $\approx 90$ depending on cap and budget, against $1.9$--$4.2$ (falling with $s$) for uniform scaling.

\begin{table*}[t]
\caption{\label{tab:h2h}Head-to-head at matched achieved error, at three confidence levels. ``total'' is $\varepsilon_\delta \Ntot$. ``max obs.'' uses chebAE's maximum observed query count over its trials, an empirical statistic rather than a worst-case bound; it is absent for the deterministic schedules, whose query count has no spread. ``seq.'' is the maximum sequential Grover depth constant, the depth a single processor must supply. We ran chebAE under our protocol with its failure parameter at $0.32$, $0.05$, and $0.01$ against our achieved $\varepsilon_\delta$, under the conventions of Ref.~\cite{labib2024csae}; csAE values are as published there. Ranges span the scales tested at each level; the upper end of our $99\%$ range is the deepest extension point of Fig.~\ref{fig:scaling} at $2\times 10^4$ trials, while the $10^6$-trial schedules of Table~\ref{tab:scale} give $5.1$--$5.8$. The flagship is competitive or better at all three confidence levels, and its sequential depth is $13\times$ smaller throughout.}
\begin{ruledtabular}
\begin{tabular}{lccccccccc}
 & \multicolumn{3}{c}{$68\%$} & \multicolumn{3}{c}{$95\%$} & \multicolumn{3}{c}{$99\%$} \\
 & total & max obs. & seq. & total & max obs. & seq. & total & max obs. & seq. \\
\hline
this work & $1.23$--$1.33$ & --- & $0.093$ & $2.76$--$3.07$ & --- & $0.21$ & $5.1$--$6.2$ & --- & $0.39$--$0.44$ \\
chebAE \cite{Rall2023amplitudeestimation} & $1.27$--$1.41$ & $1.97$--$2.43$ & $1.13$--$1.28$ & $3.00$--$3.23$ & $4.56$--$5.00$ & $2.8$--$3.0$ & $5.5$ & $8.1$ & $5.1$ \\
csAE \cite{labib2024csae} & $1.67$ & --- & $0.102$ & $4.3$ & --- & $0.26$ & $8.9$ & --- & $0.27$ \\
\end{tabular}
\end{ruledtabular}
\end{table*}

\begin{table*}[t]
\caption{\label{tab:depth}Depth-limited operation: the $r=1.45$ ladder truncated at $M$ with $M$ appended as the deepest rung, all shots scaled by a single factor $s$. Each column is one schedule; $2\times 10^4$ trials each. $\widetilde C_\delta = \varepsilon_\delta^2 M \Ntot$ is the depth--query trade-off constant; its \ac{CRLB} floor is independent of $s$ by construction, and the measured value converges to it as $s$ grows.}
\begin{ruledtabular}
\begin{tabular}{lcccccccc}
$M$ & \multicolumn{4}{c}{$64$} & \multicolumn{4}{c}{$256$} \\
$s$ & $1$ & $4$ & $16$ & $64$ & $1$ & $4$ & $16$ & $64$ \\
\hline
$\Ntot$ & $821$ & $3284$ & $13136$ & $52544$ & $2660$ & $10640$ & $42560$ & $170240$ \\
$\varepsilon_{95}$ & $3.4{\times}10^{-3}$ & $1.5{\times}10^{-3}$ & $7.0{\times}10^{-4}$ & $3.4{\times}10^{-4}$ & $1.1{\times}10^{-3}$ & $4.6{\times}10^{-4}$ & $2.2{\times}10^{-4}$ & $1.1{\times}10^{-4}$ \\
$\widetilde C_{68}$ & $0.111$ & $0.100$ & $0.090$ & $0.088$ & $0.145$ & $0.124$ & $0.117$ & $0.112$ \\
$\widetilde C_{95}$ & $0.60$ & $0.46$ & $0.41$ & $0.39$ & $0.80$ & $0.58$ & $0.52$ & $0.51$ \\
$\widetilde C_{99}$ & $2.2$ & $1.01$ & $0.81$ & $0.73$ & $4.6$ & $1.35$ & $1.07$ & $1.00$ \\
$\varepsilon_{95}/\varepsilon_{95}^{\mathrm{CRLB}}$ & $1.27$ & $1.10$ & $1.05$ & $1.02$ & $1.30$ & $1.11$ & $1.05$ & $1.04$ \\
\end{tabular}
\end{ruledtabular}
\end{table*}

\begin{figure}[t]
\centering
\includegraphics[width=0.98\columnwidth]{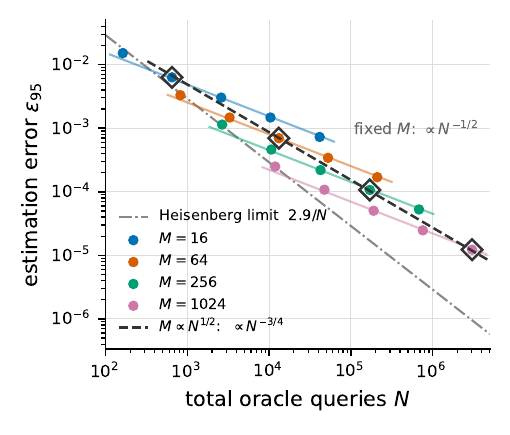}
\caption{\label{fig:fan}The depth--query trade-off as a family of power laws. Each fixed depth cap $M$ (uniform scaling, $s = 1$--$256$) follows $\varepsilon_{95} = \sqrt{\widetilde C_{95}/(M N)}$, i.e., slope $-1/2$: this is the $\beta = 0$ policy. The Heisenberg envelope $\varepsilon_{95} = 2.9/N$ (gray, slope $-1$) is the $\beta = 1$ policy, $M \propto N$. An intermediate depth-growth policy $M \propto N^{\beta}$ is a \emph{path across the fan}, switching to a deeper cap as the budget grows; the dashed line threads the four measured points with $M \propto N^{1/2}$ (open diamonds) and exhibits the predicted slope $-(1+\beta)/2 = -3/4$.}
\end{figure}

\section{Noise-aware estimation}\label{sec:noise}

The likelihood formulation extends to noise models by replacing the response function. For a depolarizing channel with per-oracle-call rate $\eta$, the fringe visibility at depth $n$ is $V_n = (1-\eta)^{2n+1}$ and
\begin{equation}\label{eq:noisymodel}
    p_n(\theta; \eta) = V_n \cos^2\!\big((2n+1)\theta\big) + \tfrac{1 - V_n}{2},
\end{equation}
a one-line change to Eq.~\eqref{eq:loglik} when $\eta$ is known (e.g., from calibration). Table~\ref{tab:noise} compares the noise-aware estimator with a mismatched estimator that ignores the noise, on the flagship schedule at two scales. Degradation is governed by the visibility of the deepest level, i.e., by the product $\eta\, n_{\max}$: for $\eta\, n_{\max} \lesssim 0.005$ the constants are essentially unchanged, at $\eta\, n_{\max} \approx 0.02$ they rise by under $10\%$, and the degradation is graceful well beyond that. Noise awareness costs nothing and matters more as $\eta$ grows: at $\eta = 10^{-3}$ per oracle call the $n_{\max} = 540$ schedule reaches $C_{95} = 7.0$ with the matched likelihood versus $10.2$ when the noise is ignored. For larger $\eta$ the correct response is to cap the ladder depth near $\sim 1/\eta$ (the likelihood then gives the resulting constant directly); a systematic study of noise-adapted ladders is left to future work.

\begin{table}[t]
\caption{\label{tab:noise}Depolarizing noise (rate $\eta$ per oracle call, Eq.~\eqref{eq:noisymodel}): $C_{95}$ of the flagship schedule with noise-aware vs.\ noise-ignorant likelihoods, $2\times 10^5$ trials.}
\begin{ruledtabular}
\begin{tabular}{lcccc}
 & \multicolumn{2}{c}{$n_{\max} = 182$} & \multicolumn{2}{c}{$n_{\max} = 540$} \\
$\eta$ & aware & ignorant & aware & ignorant \\
\hline
$0$              & $2.91$ & $2.91$ & $2.98$ & $2.98$ \\
$10^{-5}$        & $2.94$ & $2.94$ & $3.04$ & $3.04$ \\
$10^{-4}$        & $3.16$ & $3.19$ & $3.46$ & $3.56$ \\
$3\times10^{-4}$ & $3.55$ & $3.75$ & $4.19$ & $4.81$ \\
$10^{-3}$        & $4.70$ & $5.79$ & $7.03$ & $10.15$ \\
\end{tabular}
\end{ruledtabular}
\end{table}

\section{Estimator optimality beyond the Cram\'er--Rao bound}\label{sec:zzb}

The \ac{CRLB} used above is a \emph{local} bound: it presumes the estimate already lies in the correct basin, so it cannot price the ambiguity that Sec.~\ref{sec:ladder} identifies as the binding constraint. Taken literally it is also minimized by a useless schedule: a single shot at depth $n$ has $\kappa_{95}\Ntot/\sqrt{\sum_j I_{n_j}} = \kappa_{95}\,n/(2(2n+1)) < \kappa_{95}/4 \approx 0.42$, yet no estimator can extract anything from one bit, so the interesting question is not how close we sit to \emph{that} floor. The Ziv--Zakai bound introduced below assigns the same schedule (one shot at depth $125$) $C_{95} \approx 46$ (essentially the no-data value $\varepsilon_{95} \approx 0.38$ of the constant estimate $\hat a = 1/2$, times $\Ntot = 125$), and is in this sense the first quantity here that penalizes a schedule for being unidentifiable rather than only for being uninformative.

A local bound stops being achievable once well-separated rivals survive. This is the \emph{threshold effect} familiar from time-delay and direction-of-arrival estimation, and the bound built for it is due to Ziv and Zakai \cite{ziv1969some, bell1997extended}. It reduces estimation to binary hypothesis testing, which is precisely the rival-basin structure: for \emph{any} estimator,
\begin{equation}\label{eq:zzb}
    \Pr\big[|\hat a - a| \ge g/2\big] \;\ge\; \int \big[p(a) + p(a{+}g)\big]\, P_{\min}(a, a{+}g)\, da ,
\end{equation}
with $P_{\min}$ the minimum error probability of deciding between $a$ and $a+g$ from the same data. Since $|\hat a - a| \ge g'/2$ implies $|\hat a - a| \ge g/2$ for $g' \ge g$, the right-hand side may be replaced by its supremum over $g' \ge g$; the result bounds the entire error survival function, and hence \emph{every} percentile at once from a single computation. We evaluate $P_{\min}$ without inequalities rather than through the usual Bhattacharyya bound, using $\sum_x \min(p_0 P, p_1 Q) = p_0\,\mathbb{E}_{P}\big[\min(1, (p_1/p_0)e^{\Delta})\big]$ with $\Delta$ the log-likelihood ratio, an unbiased Monte Carlo average. This matters: the Bhattacharyya form is loose enough here to return a bound \emph{below} the \ac{CRLB}, and so carries no information. Appendix~\ref{app:zzb} derives Eq.~\eqref{eq:zzb}, gives the monotone-envelope (``valley-filling'') step and the reduction of every percentile to one curve, and records the numerical protocol.

\begin{table}[t]
\caption{\label{tab:zzb}How close the exact global \ac{ML} estimate is to optimal, for the flagship schedule at $n_{\max}=125$. The \ac{CRLB} is local; the Ziv--Zakai bound of Eq.~\eqref{eq:zzb} prices the ambiguity as well and applies to any estimator whatsoever. Achieved values are those of Table~\ref{tab:scale}; the achieved/Ziv--Zakai ratios at $n_{\max}=382$ are $1.04$, $1.06$, $1.35$ (to the $\approx 3\%$ accuracy of Appendix~\ref{app:zzb}).}
\begin{ruledtabular}
\begin{tabular}{lcccc}
$\delta$ & \ac{CRLB} & Ziv--Zakai & achieved & achieved/ZZ \\
\hline
$68$ & $1.135$ & $1.219$ & $1.270$ & $1.04$ \\
$95$ & $2.361$ & $2.718$ & $2.886$ & $1.06$ \\
$99$ & $3.195$ & $3.980$ & $5.324$ & $1.34$ \\
\end{tabular}
\end{ruledtabular}
\end{table}

Table~\ref{tab:zzb} gives the result. At $68\%$ and $95\%$ confidence the exact global \ac{ML} estimate sits within $4\%$ and $6\%$ of a bound no estimator can beat, at both scales tested. Most of the apparent $22\%$ shortfall against the \ac{CRLB} at $\delta = 95$ is therefore not slack but an irreducible cost of ambiguity. Little further gain is available from post-processing: whatever room remains at these confidence levels is in the \emph{schedule}, not the estimator.

At $99\%$ the picture is different, and the same analysis says why. The gap widens to $34$--$35\%$, and there the \ac{ML} estimate is no longer the right decision rule. A percentile metric is a tolerance loss, $\mathbf{1}\{|\hat a - a| > h\}$, whose Bayes rule maximizes the posterior mass \emph{within} $h$ of the estimate, that is, reports the center of the highest-mass window of half-width $h$ (the \emph{window rule}), not the mode. The \ac{ML} estimate is the $h \to 0$ member of that family (up to the prior's Jacobian, an $\mathcal{O}(\sigma_\theta^2)$ shift), which is the appropriate choice only in the limit of vanishing tolerance. Taking $h \approx 2.8\,\sigma_\theta$ (a window in $\theta$) instead lowers $C_{99}$ by $14$--$16\%$, at a cost of $\approx 2\%$ in $C_{95}$ and $\approx 9\%$ in $C_{68}$ ($2\times 10^5$ trials, both scales). We keep the plain \ac{ML} estimate as the default because it is parameter-free and because it is already near-optimal at the $95\%$ level we report. For a tail-critical application the window rule is the right choice, and it is a small change to the code.

One caveat on reading the $99\%$ row. Equation~\eqref{eq:zzb} is a \emph{pairwise} bound, testing $a$ against one rival at a time, whereas in the far tail many basins compete simultaneously; pairwise bounds are known to be loose in that regime. The $1.34$ should therefore be read as an upper bound on what any estimator could recover at $99\%$, not as an achievable target---consistent with the window rule recovering $14$--$16\%$ of it rather than all.

Neither the \ac{CRLB} nor Eq.~\eqref{eq:zzb} is a bound for the \emph{class}: both are computed for a fixed schedule. They differ, though, in whether minimizing over schedules is even meaningful. The \ac{CRLB} is not---beyond being minimized by the degenerate schedule above, it falls monotonically as the ladder is made sparser, because it prices information and ignores identifiability, so it cannot see the optimum in $r$ at all. The Ziv--Zakai bound does have an \emph{interior} minimum in $r$ over geometric ladders, reproducing on its own the Fisher-versus-aliasing trade-off that Eq.~\eqref{eq:fisherflank} captures on the small-$r$ side only. That makes $\min_{\text{schedules}}$ of Eq.~\eqref{eq:zzb}, taken at a fixed target error or in the limit $\varepsilon \to 0$ (at any fixed small budget the bound is trivially small), the natural formulation of the missing class-level lower bound on $\varepsilon_{95}\Ntot$. We do not evaluate it here, and the obstacle is the looseness just noted rather than the minimization itself: the pairwise structure loosens precisely where aliasing is severe, which is the region such a minimization would be drawn into. The substantive open problem is therefore a bound that prices many simultaneously competing rivals.

\section{Discussion}\label{sec:discussion}

Non-adaptive \ac{AE} carries two structural advantages, deterministic resource counts and embarrassingly parallel execution, that adaptive methods cannot offer, and our results show it need not pay for them in average-case total queries. The remaining gaps point to where adaptivity helps. The plain geometric ladder concedes $\sim 15\%$ to chebAE's average at $99\%$ confidence, which is what one expects if adaptivity is tail insurance: its value should grow with $\delta$, and it does. That concession is not structural. A single extra rung, which breaks the deepest level's symmetry about its fringe extrema instead of spending extra rounds on unlucky runs, closes it. At the same time it improves the parallel constant by $27\%$, something adaptivity cannot do. The residual cost is at $68\%$, where the extra rung is insurance the typical run does not need. We did not find further gains within the fully non-adaptive class. Shot reallocation across a fixed ladder, coprime and jittered ladders, and two-round semi-adaptive schemes all failed to improve on the $r \approx 1.45$ geometric ladder. (In the semi-adaptive schemes the round-two depth is capped by the round-one posterior width, which is exactly the aliasing constraint of Sec.~\ref{sec:ladder}.)

How close to optimal are these constants? The estimator is close to optimal: Sec.~\ref{sec:zzb} places the exact global \ac{ML} estimate within $4\%$ and $6\%$ of a Ziv--Zakai bound at $68\%$ and $95\%$ confidence ($\varepsilon_{95}\, \Ntot \approx 2.9$ achieved against a bound of $\approx 2.7$). Measured against the schedule's \ac{CRLB} the gap looks like $\approx 22\%$, but most of that is the \ac{CRLB}'s blindness to ambiguity rather than anything an estimator could collect. What room remains is in the \emph{schedule}; the class-level bound that would price it is still missing (Sec.~\ref{sec:zzb}), and our schedule searches suggest the achievable value is not far below $2.8$.

For the parallel direction the accounting is different. No bound on $\varepsilon_{95}\, n_L$ alone exists, since at a fixed depth more shots always buy more precision (the fan of Fig.~\ref{fig:fan} makes this visible). The bounded object is the \emph{product} $\widetilde C = (\varepsilon_{95}\Ntot)(\varepsilon_{95} n_L)$, and there uniform scaling sits within $\sim 1.1\times$ of its ladder-shape floor ($\approx 0.37$--$0.47$). The absolute floor ($\approx 0.17$, all information at the cap) is approachable at the $95$th percentile by anchored-pair schedules, but only at the price of a polynomially decaying catastrophic tail (Sec.~\ref{sec:depth}). The $\sim 2\times$ premium of uniform scaling is what exponential, budget-uniform tail control costs in this class.

On the depth axis, Sec.~\ref{sec:depth} supplies explicit constants for the depth-limited trade-off established asymptotically in Refs.~\cite{giurgica2022low, huang2026eigengap, erle2026anydepth}; logarithmic-depth \ac{AE} via distributed GHZ states \cite{oshio2025near} occupies a further regime requiring inter-processor entanglement, a resource the present method does not use. The uniform-in-angle constructions of Ref.~\cite{erle2026anydepth} are a natural route to lifting the amplitude-range restriction below.

All results use the standard shot-noise model at $a \in (0.1, 0.9)$, matching the conventions of Refs.~\cite{suzuki2020amplitude, grinko2021iterative, Rall2023amplitudeestimation, labib2024csae}. The results do not depend on that convention: widening the prior to $\mathcal{U}(0.01, 0.99)$ moves $C_{95}$ by under $2\%$. The direction of the change matters. Because the $\cos\theta$ Jacobian compresses the additive error as $a \to 1$, the wider prior is slightly \emph{easier} for both methods. Being adaptive, chebAE gains more from this, so our margin over it narrows rather than widens.

Every claim above is internal to the class of incoherent Grover circuits with $Z$-basis measurements: strategies that spend queries coherently---\ac{QPE}-type algorithms or distributed-entanglement schemes \cite{oshio2025near}---are bounded differently and can in principle improve the constants further, at the hardware cost this class is designed to avoid.

\begin{acknowledgments}
The author thanks the co-authors of Ref.~\cite{labib2024csae} for the framework this work builds on. Numerical exploration, optimization, and drafting were carried out with substantial assistance from Claude (Anthropic).
\end{acknowledgments}

\bibliography{refs}

\begin{thebibliography}{19}%
\makeatletter
\providecommand \@ifxundefined [1]{%
 \@ifx{#1\undefined}
}%
\providecommand \@ifnum [1]{%
 \ifnum #1\expandafter \@firstoftwo
 \else \expandafter \@secondoftwo
 \fi
}%
\providecommand \@ifx [1]{%
 \ifx #1\expandafter \@firstoftwo
 \else \expandafter \@secondoftwo
 \fi
}%
\providecommand \natexlab [1]{#1}%
\providecommand \enquote  [1]{``#1''}%
\providecommand \bibnamefont  [1]{#1}%
\providecommand \bibfnamefont [1]{#1}%
\providecommand \citenamefont [1]{#1}%
\providecommand \href@noop [0]{\@secondoftwo}%
\providecommand \href [0]{\begingroup \@sanitize@url \@href}%
\providecommand \@href[1]{\@@startlink{#1}\@@href}%
\providecommand \@@href[1]{\endgroup#1\@@endlink}%
\providecommand \@sanitize@url [0]{\catcode `\\12\catcode `\$12\catcode
  `\&12\catcode `\#12\catcode `\^12\catcode `\_12\catcode `\%12\relax}%
\providecommand \@@startlink[1]{}%
\providecommand \@@endlink[0]{}%
\providecommand \url  [0]{\begingroup\@sanitize@url \@url }%
\providecommand \@url [1]{\endgroup\@href {#1}{\urlprefix }}%
\providecommand \urlprefix  [0]{URL }%
\providecommand \Eprint [0]{\href }%
\providecommand \doibase [0]{https://doi.org/}%
\providecommand \selectlanguage [0]{\@gobble}%
\providecommand \bibinfo  [0]{\@secondoftwo}%
\providecommand \bibfield  [0]{\@secondoftwo}%
\providecommand \translation [1]{[#1]}%
\providecommand \BibitemOpen [0]{}%
\providecommand \bibitemStop [0]{}%
\providecommand \bibitemNoStop [0]{.\EOS\space}%
\providecommand \EOS [0]{\spacefactor3000\relax}%
\providecommand \BibitemShut  [1]{\csname bibitem#1\endcsname}%
\let\auto@bib@innerbib\@empty
\bibitem [{\citenamefont {Brassard}\ \emph {et~al.}(2002)\citenamefont
  {Brassard}, \citenamefont {H{\o}yer}, \citenamefont {Mosca},\ and\
  \citenamefont {Tapp}}]{brassard2002quantum}%
  \BibitemOpen
  \bibfield  {author} {\bibinfo {author} {\bibfnamefont {G.}~\bibnamefont
  {Brassard}}, \bibinfo {author} {\bibfnamefont {P.}~\bibnamefont {H{\o}yer}},
  \bibinfo {author} {\bibfnamefont {M.}~\bibnamefont {Mosca}},\ and\ \bibinfo
  {author} {\bibfnamefont {A.}~\bibnamefont {Tapp}},\ }\bibfield  {title}
  {\bibinfo {title} {Quantum amplitude amplification and estimation},\ }\href
  {https://doi.org/10.1090/conm/305/05215} {\bibfield  {journal} {\bibinfo
  {journal} {Contemporary Mathematics}\ }\textbf {\bibinfo {volume} {305}},\
  \bibinfo {pages} {53} (\bibinfo {year} {2002})}\BibitemShut {NoStop}%
\bibitem [{\citenamefont {Montanaro}(2015)}]{montanaro2015quantum}%
  \BibitemOpen
  \bibfield  {author} {\bibinfo {author} {\bibfnamefont {A.}~\bibnamefont
  {Montanaro}},\ }\bibfield  {title} {\bibinfo {title} {Quantum speedup of
  {M}onte {C}arlo methods},\ }\href {https://doi.org/10.1098/rspa.2015.0301}
  {\bibfield  {journal} {\bibinfo  {journal} {Proceedings of the Royal Society
  A: Mathematical, Physical and Engineering Sciences}\ }\textbf {\bibinfo
  {volume} {471}},\ \bibinfo {pages} {20150301} (\bibinfo {year}
  {2015})}\BibitemShut {NoStop}%
\bibitem [{\citenamefont {Rebentrost}\ \emph {et~al.}(2018)\citenamefont
  {Rebentrost}, \citenamefont {Gupt},\ and\ \citenamefont
  {Bromley}}]{rebentrost2018quantum}%
  \BibitemOpen
  \bibfield  {author} {\bibinfo {author} {\bibfnamefont {P.}~\bibnamefont
  {Rebentrost}}, \bibinfo {author} {\bibfnamefont {B.}~\bibnamefont {Gupt}},\
  and\ \bibinfo {author} {\bibfnamefont {T.~R.}\ \bibnamefont {Bromley}},\
  }\bibfield  {title} {\bibinfo {title} {Quantum computational finance: {M}onte
  {C}arlo pricing of financial derivatives},\ }\href
  {https://doi.org/10.1103/PhysRevA.98.022321} {\bibfield  {journal} {\bibinfo
  {journal} {Phys. Rev. A}\ }\textbf {\bibinfo {volume} {98}},\ \bibinfo
  {pages} {022321} (\bibinfo {year} {2018})}\BibitemShut {NoStop}%
\bibitem [{\citenamefont {Stamatopoulos}\ \emph {et~al.}(2020)\citenamefont
  {Stamatopoulos}, \citenamefont {Egger}, \citenamefont {Sun}, \citenamefont
  {Zoufal}, \citenamefont {Iten}, \citenamefont {Shen},\ and\ \citenamefont
  {Woerner}}]{Stamatopoulos2020optionpricingusing}%
  \BibitemOpen
  \bibfield  {author} {\bibinfo {author} {\bibfnamefont {N.}~\bibnamefont
  {Stamatopoulos}}, \bibinfo {author} {\bibfnamefont {D.~J.}\ \bibnamefont
  {Egger}}, \bibinfo {author} {\bibfnamefont {Y.}~\bibnamefont {Sun}}, \bibinfo
  {author} {\bibfnamefont {C.}~\bibnamefont {Zoufal}}, \bibinfo {author}
  {\bibfnamefont {R.}~\bibnamefont {Iten}}, \bibinfo {author} {\bibfnamefont
  {N.}~\bibnamefont {Shen}},\ and\ \bibinfo {author} {\bibfnamefont
  {S.}~\bibnamefont {Woerner}},\ }\bibfield  {title} {\bibinfo {title} {Option
  {P}ricing using {Q}uantum {C}omputers},\ }\href
  {https://doi.org/10.22331/q-2020-07-06-291} {\bibfield  {journal} {\bibinfo
  {journal} {{Quantum}}\ }\textbf {\bibinfo {volume} {4}},\ \bibinfo {pages}
  {291} (\bibinfo {year} {2020})}\BibitemShut {NoStop}%
\bibitem [{\citenamefont {Chakrabarti}\ \emph {et~al.}(2021)\citenamefont
  {Chakrabarti}, \citenamefont {Krishnakumar}, \citenamefont {Mazzola},
  \citenamefont {Stamatopoulos}, \citenamefont {Woerner},\ and\ \citenamefont
  {Zeng}}]{Chakrabarti2021thresholdquantum}%
  \BibitemOpen
  \bibfield  {author} {\bibinfo {author} {\bibfnamefont {S.}~\bibnamefont
  {Chakrabarti}}, \bibinfo {author} {\bibfnamefont {R.}~\bibnamefont
  {Krishnakumar}}, \bibinfo {author} {\bibfnamefont {G.}~\bibnamefont
  {Mazzola}}, \bibinfo {author} {\bibfnamefont {N.}~\bibnamefont
  {Stamatopoulos}}, \bibinfo {author} {\bibfnamefont {S.}~\bibnamefont
  {Woerner}},\ and\ \bibinfo {author} {\bibfnamefont {W.~J.}\ \bibnamefont
  {Zeng}},\ }\bibfield  {title} {\bibinfo {title} {A {T}hreshold for {Q}uantum
  {A}dvantage in {D}erivative {P}ricing},\ }\href
  {https://doi.org/10.22331/q-2021-06-01-463} {\bibfield  {journal} {\bibinfo
  {journal} {{Quantum}}\ }\textbf {\bibinfo {volume} {5}},\ \bibinfo {pages}
  {463} (\bibinfo {year} {2021})}\BibitemShut {NoStop}%
\bibitem [{\citenamefont {Suzuki}\ \emph {et~al.}(2020)\citenamefont {Suzuki},
  \citenamefont {Uno}, \citenamefont {Raymond}, \citenamefont {Tanaka},
  \citenamefont {Onodera},\ and\ \citenamefont
  {Yamamoto}}]{suzuki2020amplitude}%
  \BibitemOpen
  \bibfield  {author} {\bibinfo {author} {\bibfnamefont {Y.}~\bibnamefont
  {Suzuki}}, \bibinfo {author} {\bibfnamefont {S.}~\bibnamefont {Uno}},
  \bibinfo {author} {\bibfnamefont {R.}~\bibnamefont {Raymond}}, \bibinfo
  {author} {\bibfnamefont {T.}~\bibnamefont {Tanaka}}, \bibinfo {author}
  {\bibfnamefont {T.}~\bibnamefont {Onodera}},\ and\ \bibinfo {author}
  {\bibfnamefont {N.}~\bibnamefont {Yamamoto}},\ }\bibfield  {title} {\bibinfo
  {title} {Amplitude estimation without phase estimation},\ }\href
  {https://doi.org/10.1007/s11128-019-2565-2} {\bibfield  {journal} {\bibinfo
  {journal} {Quantum Information Processing}\ }\textbf {\bibinfo {volume}
  {19}},\ \bibinfo {pages} {75} (\bibinfo {year} {2020})}\BibitemShut {NoStop}%
\bibitem [{\citenamefont {Aaronson}\ and\ \citenamefont
  {Rall}(2020)}]{aaronson2020quantum}%
  \BibitemOpen
  \bibfield  {author} {\bibinfo {author} {\bibfnamefont {S.}~\bibnamefont
  {Aaronson}}\ and\ \bibinfo {author} {\bibfnamefont {P.}~\bibnamefont
  {Rall}},\ }\bibfield  {title} {\bibinfo {title} {Quantum approximate
  counting, simplified},\ }in\ \href
  {https://doi.org/10.1137/1.9781611976014.5} {\emph {\bibinfo {booktitle}
  {Symposium on {S}implicity in {A}lgorithms (SOSA)}}}\ (\bibinfo
  {organization} {SIAM},\ \bibinfo {year} {2020})\ pp.\ \bibinfo {pages}
  {24--32}\BibitemShut {NoStop}%
\bibitem [{\citenamefont {Grinko}\ \emph {et~al.}(2021)\citenamefont {Grinko},
  \citenamefont {Gacon}, \citenamefont {Zoufal},\ and\ \citenamefont
  {Woerner}}]{grinko2021iterative}%
  \BibitemOpen
  \bibfield  {author} {\bibinfo {author} {\bibfnamefont {D.}~\bibnamefont
  {Grinko}}, \bibinfo {author} {\bibfnamefont {J.}~\bibnamefont {Gacon}},
  \bibinfo {author} {\bibfnamefont {C.}~\bibnamefont {Zoufal}},\ and\ \bibinfo
  {author} {\bibfnamefont {S.}~\bibnamefont {Woerner}},\ }\bibfield  {title}
  {\bibinfo {title} {Iterative quantum amplitude estimation},\ }\href
  {https://doi.org/10.1038/s41534-021-00379-1} {\bibfield  {journal} {\bibinfo
  {journal} {npj Quantum Information}\ }\textbf {\bibinfo {volume} {7}},\
  \bibinfo {pages} {52} (\bibinfo {year} {2021})}\BibitemShut {NoStop}%
\bibitem [{\citenamefont {Venkateswaran}\ and\ \citenamefont
  {O'Donnell}(2021)}]{venkateswaran2021quantum}%
  \BibitemOpen
  \bibfield  {author} {\bibinfo {author} {\bibfnamefont {R.}~\bibnamefont
  {Venkateswaran}}\ and\ \bibinfo {author} {\bibfnamefont {R.}~\bibnamefont
  {O'Donnell}},\ }\bibfield  {title} {\bibinfo {title} {Quantum approximate
  counting with nonadaptive {G}rover iterations},\ }in\ \href
  {https://doi.org/10.4230/LIPIcs.STACS.2021.59} {\emph {\bibinfo {booktitle}
  {38th International Symposium on Theoretical Aspects of Computer Science
  (STACS 2021)}}},\ \bibinfo {series} {LIPIcs}, Vol.\ \bibinfo {volume} {187}\
  (\bibinfo {year} {2021})\ pp.\ \bibinfo {pages} {59:1--59:12},\ \Eprint
  {https://arxiv.org/abs/2010.04370} {arXiv:2010.04370 [quant-ph]} \BibitemShut
  {NoStop}%
\bibitem [{\citenamefont {Giurgica-Tiron}\ \emph {et~al.}(2022)\citenamefont
  {Giurgica-Tiron}, \citenamefont {Kerenidis}, \citenamefont {Labib},
  \citenamefont {Prakash},\ and\ \citenamefont {Zeng}}]{giurgica2022low}%
  \BibitemOpen
  \bibfield  {author} {\bibinfo {author} {\bibfnamefont {T.}~\bibnamefont
  {Giurgica-Tiron}}, \bibinfo {author} {\bibfnamefont {I.}~\bibnamefont
  {Kerenidis}}, \bibinfo {author} {\bibfnamefont {F.}~\bibnamefont {Labib}},
  \bibinfo {author} {\bibfnamefont {A.}~\bibnamefont {Prakash}},\ and\ \bibinfo
  {author} {\bibfnamefont {W.}~\bibnamefont {Zeng}},\ }\bibfield  {title}
  {\bibinfo {title} {Low depth algorithms for quantum amplitude estimation},\
  }\href {https://doi.org/10.22331/q-2022-06-27-745} {\bibfield  {journal}
  {\bibinfo  {journal} {Quantum}\ }\textbf {\bibinfo {volume} {6}},\ \bibinfo
  {pages} {745} (\bibinfo {year} {2022})}\BibitemShut {NoStop}%
\bibitem [{\citenamefont {Rall}\ and\ \citenamefont
  {Fuller}(2023)}]{Rall2023amplitudeestimation}%
  \BibitemOpen
  \bibfield  {author} {\bibinfo {author} {\bibfnamefont {P.}~\bibnamefont
  {Rall}}\ and\ \bibinfo {author} {\bibfnamefont {B.}~\bibnamefont {Fuller}},\
  }\bibfield  {title} {\bibinfo {title} {Amplitude {E}stimation from {Q}uantum
  {S}ignal {P}rocessing},\ }\href {https://doi.org/10.22331/q-2023-03-02-937}
  {\bibfield  {journal} {\bibinfo  {journal} {{Quantum}}\ }\textbf {\bibinfo
  {volume} {7}},\ \bibinfo {pages} {937} (\bibinfo {year} {2023})}\BibitemShut
  {NoStop}%
\bibitem [{\citenamefont {Ram{\^o}a}\ and\ \citenamefont
  {Santos}(2025)}]{ramoa2025bayesian}%
  \BibitemOpen
  \bibfield  {author} {\bibinfo {author} {\bibfnamefont {A.}~\bibnamefont
  {Ram{\^o}a}}\ and\ \bibinfo {author} {\bibfnamefont {L.~P.}\ \bibnamefont
  {Santos}},\ }\bibfield  {title} {\bibinfo {title} {Bayesian quantum amplitude
  estimation},\ }\href {https://doi.org/10.22331/q-2025-09-11-1856} {\bibfield
  {journal} {\bibinfo  {journal} {Quantum}\ }\textbf {\bibinfo {volume} {9}},\
  \bibinfo {pages} {1856} (\bibinfo {year} {2025})}\BibitemShut {NoStop}%
\bibitem [{\citenamefont {Li}\ \emph {et~al.}(2026)\citenamefont {Li},
  \citenamefont {Vidwans}, \citenamefont {Wang},\ and\ \citenamefont
  {Soley}}]{li2026biqae}%
  \BibitemOpen
  \bibfield  {author} {\bibinfo {author} {\bibfnamefont {Q.}~\bibnamefont
  {Li}}, \bibinfo {author} {\bibfnamefont {A.}~\bibnamefont {Vidwans}},
  \bibinfo {author} {\bibfnamefont {Y.}~\bibnamefont {Wang}},\ and\ \bibinfo
  {author} {\bibfnamefont {M.~B.}\ \bibnamefont {Soley}},\ }\bibfield  {title}
  {\bibinfo {title} {Harnessing {B}ayesian statistics to accelerate iterative
  quantum amplitude estimation},\ }\href
  {https://doi.org/10.22331/q-2026-01-14-1962} {\bibfield  {journal} {\bibinfo
  {journal} {Quantum}\ }\textbf {\bibinfo {volume} {10}},\ \bibinfo {pages}
  {1962} (\bibinfo {year} {2026})}\BibitemShut {NoStop}%
\bibitem [{\citenamefont {Labib}\ \emph {et~al.}(2024)\citenamefont {Labib},
  \citenamefont {Clader}, \citenamefont {Stamatopoulos},\ and\ \citenamefont
  {Zeng}}]{labib2024csae}%
  \BibitemOpen
  \bibfield  {author} {\bibinfo {author} {\bibfnamefont {F.}~\bibnamefont
  {Labib}}, \bibinfo {author} {\bibfnamefont {B.~D.}\ \bibnamefont {Clader}},
  \bibinfo {author} {\bibfnamefont {N.}~\bibnamefont {Stamatopoulos}},\ and\
  \bibinfo {author} {\bibfnamefont {W.~J.}\ \bibnamefont {Zeng}},\ }\href@noop
  {} {\bibinfo {title} {Quantum amplitude estimation from classical signal
  processing}} (\bibinfo {year} {2024}),\ \Eprint
  {https://arxiv.org/abs/2405.14697} {arXiv:2405.14697 [quant-ph]} \BibitemShut
  {NoStop}%
\bibitem [{\citenamefont {Huang}\ and\ \citenamefont
  {Koczor}(2026)}]{huang2026eigengap}%
  \BibitemOpen
  \bibfield  {author} {\bibinfo {author} {\bibfnamefont {P.-W.}\ \bibnamefont
  {Huang}}\ and\ \bibinfo {author} {\bibfnamefont {B.}~\bibnamefont {Koczor}},\
  }\href@noop {} {\bibinfo {title} {Low-depth amplitude estimation via
  statistical eigengap estimation}} (\bibinfo {year} {2026}),\ \Eprint
  {https://arxiv.org/abs/2603.05475} {arXiv:2603.05475 [quant-ph]} \BibitemShut
  {NoStop}%
\bibitem [{\citenamefont {Erle}\ and\ \citenamefont
  {Koczor}(2026)}]{erle2026anydepth}%
  \BibitemOpen
  \bibfield  {author} {\bibinfo {author} {\bibfnamefont {J.}~\bibnamefont
  {Erle}}\ and\ \bibinfo {author} {\bibfnamefont {B.}~\bibnamefont {Koczor}},\
  }\href@noop {} {\bibinfo {title} {Nearly optimal amplitude estimation at any
  depth}} (\bibinfo {year} {2026}),\ \Eprint {https://arxiv.org/abs/2608.24434}
  {arXiv:2608.24434 [quant-ph]} \BibitemShut {NoStop}%
\bibitem [{\citenamefont {Ziv}\ and\ \citenamefont
  {Zakai}(1969)}]{ziv1969some}%
  \BibitemOpen
  \bibfield  {author} {\bibinfo {author} {\bibfnamefont {J.}~\bibnamefont
  {Ziv}}\ and\ \bibinfo {author} {\bibfnamefont {M.}~\bibnamefont {Zakai}},\
  }\bibfield  {title} {\bibinfo {title} {Some lower bounds on signal parameter
  estimation},\ }\href {https://doi.org/10.1109/TIT.1969.1054301} {\bibfield
  {journal} {\bibinfo  {journal} {IEEE Transactions on Information Theory}\
  }\textbf {\bibinfo {volume} {15}},\ \bibinfo {pages} {386} (\bibinfo {year}
  {1969})}\BibitemShut {NoStop}%
\bibitem [{\citenamefont {Bell}\ \emph {et~al.}(1997)\citenamefont {Bell},
  \citenamefont {Steinberg}, \citenamefont {Ephraim},\ and\ \citenamefont
  {Van~Trees}}]{bell1997extended}%
  \BibitemOpen
  \bibfield  {author} {\bibinfo {author} {\bibfnamefont {K.~L.}\ \bibnamefont
  {Bell}}, \bibinfo {author} {\bibfnamefont {Y.}~\bibnamefont {Steinberg}},
  \bibinfo {author} {\bibfnamefont {Y.}~\bibnamefont {Ephraim}},\ and\ \bibinfo
  {author} {\bibfnamefont {H.~L.}\ \bibnamefont {Van~Trees}},\ }\bibfield
  {title} {\bibinfo {title} {Extended {Ziv--Zakai} lower bound for vector
  parameter estimation},\ }\href {https://doi.org/10.1109/18.556118} {\bibfield
   {journal} {\bibinfo  {journal} {IEEE Transactions on Information Theory}\
  }\textbf {\bibinfo {volume} {43}},\ \bibinfo {pages} {624} (\bibinfo {year}
  {1997})}\BibitemShut {NoStop}%
\bibitem [{\citenamefont {Oshio}\ \emph {et~al.}(2025)\citenamefont {Oshio},
  \citenamefont {Wada},\ and\ \citenamefont {Yamamoto}}]{oshio2025near}%
  \BibitemOpen
  \bibfield  {author} {\bibinfo {author} {\bibfnamefont {K.}~\bibnamefont
  {Oshio}}, \bibinfo {author} {\bibfnamefont {K.}~\bibnamefont {Wada}},\ and\
  \bibinfo {author} {\bibfnamefont {N.}~\bibnamefont {Yamamoto}},\ }\href@noop
  {} {\bibinfo {title} {Near-{H}eisenberg-limited parallel amplitude estimation
  with logarithmic depth circuit}} (\bibinfo {year} {2025}),\ \Eprint
  {https://arxiv.org/abs/2508.06121} {arXiv:2508.06121 [quant-ph]} \BibitemShut
  {NoStop}%
\end{thebibliography}%

\appendix

\section{The exact global maximum-likelihood estimator}\label{app:ml}

Section~\ref{sec:ml} states that the global maximizer of Eq.~\eqref{eq:loglik} can be found exactly by a grid search plus local refinement. This appendix supplies the analytic structure behind that claim, the resolution rule that makes it hold uniformly across the budget range of Sec.~\ref{sec:depth}, the cost accounting, and the numerical check of both. Throughout, $d_j := 2n_j + 1$ is the number of oracle calls in a depth-$n_j$ circuit, $d_L$ its value at the deepest rung, and $k_j \sim \mathrm{Bin}(N_j, p_{n_j}(\theta))$ the observed counts.

\subsection{Analytic structure of the likelihood}\label{app:structure}

The likelihood is
\begin{equation}\label{eq:lik}
    \mathcal{L}(\theta) \;=\; \prod_j p_{n_j}(\theta)^{k_j}\,\big[1 - p_{n_j}(\theta)\big]^{N_j - k_j}.
\end{equation}
Since $p_n(\theta) = \tfrac12\big[1 + \cos(2 d_n \theta)\big]$ and $1 - p_n(\theta) = \tfrac12\big[1 - \cos(2 d_n\theta)\big]$, every factor is a trigonometric polynomial of degree $2d_j$, so $\mathcal{L}$ is a trigonometric polynomial of degree
\begin{equation}\label{eq:degree}
    D \;=\; 2\sum_j N_j d_j \;=\; 4\big(\Ntot - N_1\big) + 2\sum_j N_j \;\approx\; 4\,\Ntot ,
\end{equation}
even in $\theta$ and $\pi$-periodic; $(0, \pi/2)$ is a fundamental domain, which is why the search range needs no further justification.

Two points need care.

First, the \emph{log}-likelihood $\ell = \ln \mathcal{L}$ is not a trigonometric polynomial. Wherever $p_{n_j}$ vanishes with $k_j > 0$ (or $1 - p_{n_j}$ vanishes with $k_j < N_j$), $\ell$ has a logarithmic singularity to $-\infty$. These are downward spikes and never maxima, so they cannot conceal the argmax, but they do preclude uniform Lipschitz or Bernstein-type control of $\ell$; analytic statements are therefore made about $\mathcal{L}$, and the numerical work uses $\ell$ only through values at sampled points.

Second, the degree \eqref{eq:degree} is a poor guide to the required grid. Determining a degree-$D$ trigonometric polynomial from samples needs $2D + 1 \approx 8\,\Ntot$ points per period, whereas the grid used below has $16 d_L \approx 32\, n_L$ points per period; since $\Ntot \approx 10\, n_L$ for the plain ladder, a Nyquist argument falls short by a factor $\approx 2.5$ and cannot justify the search. What controls the grid is not the highest harmonic present in $\mathcal{L}$ but the width of its \emph{modes}, which is far larger.

\subsection{Two length scales: basins and modes}\label{app:scales}

Two scales govern the problem, and the resolution analysis rests on keeping them apart.

The \emph{basin width} $w = \pi/(2 d_L)$ is set by the deepest rung alone: it is the spacing between the rival solutions that the deepest fringe cannot distinguish, and hence the scale on which the coarse grid must not skip a candidate.

The \emph{mode width} is a property of the entire schedule. Near a local maximum the curvature of $\ell$ is the observed information, whose expectation at the truth is the Fisher information of Lemma~\ref{lem:fisher}; a mode therefore has Gaussian standard deviation $\sigma_\theta = \big(\sum_j 4 N_j d_j^2\big)^{-1/2}$. The ratio of the two scales depends only on the ladder shape and shot profile,
\begin{equation}\label{eq:modewidth}
    \frac{w}{\sigma_\theta} \;=\; \pi\sqrt{\Lambda}, \qquad
    \Lambda \;:=\; \sum_j N_j \left(\frac{d_j}{d_L}\right)^{\!2},
\end{equation}
and is independent of $n_L$. For the plain geometric ladder of Sec.~\ref{sec:ladder} ($r = 1.45$, canonical shots) the sum gives $\Lambda \approx 3.68$ and $w/\sigma_\theta \approx 6.03$: a likelihood mode is about one sixth of a basin wide. The flagship, carrying the extra rung, has $\Lambda \approx 5.78$ and $w/\sigma_\theta \approx 7.55$. We use the plain ladder as the running example throughout this appendix, since it is the object Sec.~\ref{sec:ladder} analyzes and none of the conclusions depend on the choice. Under the uniform scaling of Sec.~\ref{sec:depth}, $N_j \to s N_j$ sends $\Lambda \to s\Lambda$, so
\begin{equation}\label{eq:modescale}
    w/\sigma_\theta \;\approx\; 6\sqrt{s}.
\end{equation}
Equation~\eqref{eq:modewidth} is an identity; its content is the underlying claim that a mode of $\ell$ really is $\sigma_\theta$ wide. Measuring the curvature of $\ell$ at its peak confirms this (Table~\ref{tab:grid}, columns 2--3), with one caveat that the grid rule below must respect. At $s = 1$ the observed information fluctuates substantially from trial to trial, so individual modes run both wider and narrower than $\sigma_\theta$. The grid has to accommodate the narrow tail, not the median.

\subsection{The estimator and its resolution rule}\label{app:algorithm}

The estimator is lines 5--11 of Algorithm~\ref{alg:method}; what remains is to fix $\Delta$ and to say why the procedure is correct. The zoom and interpolation exist because the coarse grid is chosen to separate \emph{basins}, not to resolve within one; without them the returned value would be quantized at $\Delta/2$, which becomes the dominant error as soon as $\sigma_\theta$ falls below the coarse step. That is exactly what happens under uniform scaling, and it is the reason the rule for $\Delta$ carries two terms:
\begin{equation}\label{eq:gridrule}
    \Delta \;=\; \min\!\Big(\tfrac{w}{8},\; 4\,\sigma_\theta\Big).
\end{equation}
The first term places $\approx 8$ samples across every basin, so that no rival mode is missed by the coarse stage. The second keeps the coarse step below $4\sigma_\theta$, so that the zoom window $\pm\Delta/2 = \pm 2\sigma_\theta$, where the refinement looks, is guaranteed to straddle the mode the coarse stage selected. The zoom samples $\ell$ at $25$ points across that window, so its spacing is at most $4\sigma_\theta/24 = \sigma_\theta/6$, well inside the quadratic region, and the parabola through the best three points locates the maximum to far below $\sigma_\theta$; the count is a convenience, not a tuned value, and any count above about a dozen gives the same estimates. Both quantities are known before any data are taken ($w$ from the deepest depth, $\sigma_\theta$ from Lemma~\ref{lem:fisher}), so \eqref{eq:gridrule} involves no tuning and no data-dependent branching. By \eqref{eq:modescale} the first term binds while $s \lesssim 30$; the second binds for the more heavily scaled depth-limited designs, where it costs only a proportionally finer grid.

Table~\ref{tab:grid} shows the consequence on the $n_{\max} = 125$ ladder, $2\times10^4$ paired trials per row, against a reference search at $\Delta = w/512$. Rule \eqref{eq:gridrule} holds the error percentile within $0.1\%$ of the reference at every scale, whereas the single-term choice $\Delta = w/8$ drifts once the mode falls well inside the coarse step: $0.3\%$ at $s = 64$, $1.2\%$ at $s = 256$, $4.2\%$ at $s = 1024$. The residual $0.3\%$ at $s = 64$ is the only place where this matters for the numbers reported above---it is the largest scaling in Table~\ref{tab:depth}, and it is an order of magnitude below the two-digit precision quoted there.

\begin{table}[t]
\caption{\label{tab:grid}Mode width and grid resolution on the plain $n_{\max}=125$, $r=1.45$ ladder (Sec.~\ref{sec:ladder}) with all shots scaled by $s$. Columns 2--3: the predicted ratio of basin to mode width, Eq.~\eqref{eq:modewidth}, against the median measured from the curvature of $\ell$ at its peak ($4\times10^3$ trials). Columns 4--5: $\varepsilon_{95}$ relative to a reference search at $\Delta = w/512$, for the two grid rules ($2\times 10^4$ paired trials, identical measurement records across columns).}
\begin{ruledtabular}
\begin{tabular}{rcccc}
 & \multicolumn{2}{c}{$w/(\text{mode width})$} & \multicolumn{2}{c}{$\varepsilon_{95}$ vs.\ reference} \\
$s$ & $\pi\sqrt\Lambda$ & measured & $\Delta = w/8$ & Eq.~\eqref{eq:gridrule} \\
\hline
$1$    & $6.0$   & $5.4$   & $-0.05\%$ & $-0.05\%$ \\
$16$   & $24.1$  & $23.4$  & $-0.06\%$ & $-0.06\%$ \\
$64$   & $48.2$  & $47.8$  & $+0.28\%$ & $+0.04\%$ \\
$256$  & $96.5$  & $96.1$  & $+1.17\%$ & $-0.01\%$ \\
$1024$ & $192.9$ & $192.7$ & $+4.21\%$ & $-0.02\%$ \\
\end{tabular}
\end{ruledtabular}
\end{table}

Finally, the reason one fixed rule suffices at every budget. The coarse stage can fail only by returning a point in the wrong basin, and only if the likelihood deficit it incurs by sampling the true mode off-center exceeds the margin by which the true basin beats its rivals. The deficit is bounded by the local quadratic expansion---the nearest grid point lies within $\Delta/2$ of the mode center---giving
\begin{equation}\label{eq:deficit}
    \delta_{\mathrm{grid}} \;=\; \frac{(\Delta/2)^2}{2\sigma_\theta^2} \;=\; \frac{\pi^2 \Lambda}{512}
\end{equation}
nats when the first term of \eqref{eq:gridrule} binds, against the rival-rejection exponents $E$ of Eq.~\eqref{eq:chernoff}. Both $\delta_{\mathrm{grid}}$ and $E$ are linear in the shot counts, so the margin $E/\delta_{\mathrm{grid}}$ is \emph{invariant under uniform scaling}: a rule verified at one budget remains valid at every other. This is a scale argument rather than a theorem, since $\delta_{\mathrm{grid}}$ holds only in the quadratic regime and $E$ is itself a bound, so we check it directly against a $32\times$ finer coarse search on identical records. Over three ladders at $4\times10^4$ trials each, the two searches disagree about the basin on $5$--$7 \times 10^{-4}$ of trials. On every one of those the two candidates lie within $0.07$ nats of each other: these are statistical near-ties in which the data do not prefer either basin, not cases where the coarse grid missed a clear maximum. The error percentiles correspondingly agree to within $0.05\%$, far inside the bootstrap intervals of Table~\ref{tab:scale}.

\subsection{Cost}\label{app:cost}

With $\Delta = w/8$ the grid holds $G = (\pi/2)/\Delta = 8 d_L \approx 16\, n_L$ points, and the coarse stage costs two $G \times L$ matrix--vector products per estimate, or $\approx 4 G L \approx 64\, L\, n_L$ floating-point operations; the zoom and interpolation stages add $\mathcal{O}(L)$ and are negligible. This is the $\mathcal{O}(L\, n_L)$ quoted in Sec.~\ref{sec:ml}. Measured single-machine timings (4 BLAS threads) are $0.02$~ms per estimate at $n_{\max} = 125$, $0.10$~ms at $n_{\max} = 1688$ (the deepest row of Table~\ref{tab:scale}), and $0.77$~ms at $n_{\max} = 2\times10^4$; extrapolating the same accounting to the deepest schedule of Fig.~\ref{fig:scaling} ($n_{\max} \approx 2.1\times10^5$, $L = 35$) gives $\approx 5\times10^8$ operations and $\approx 10$~ms per estimate. The $(G \times L)$ probability matrix is built once per schedule and reused across all trials; trials are processed in chunks so that the $(\text{chunk} \times G)$ intermediate fits in memory, which is the only implementation constraint at large $n_L$.

For comparison, the ESPRIT pipeline it replaces requires, per estimate, the construction of a virtual array, a covariance estimate, an eigendecomposition, and a search over sign assignments---asymptotically and practically more expensive, and with failure modes (rank selection, sign errors) of its own.

\subsection{Noise, and the meaning of ``exact''}\label{app:caveats}

Because the estimator consumes only the model probabilities $p_{n_j}(\theta)$, any response function can be substituted without touching the algorithm: for the depolarizing model of Sec.~\ref{sec:noise} one replaces $P_{gj}$ by $V_j \cos^2(d_j\theta_g) + (1 - V_j)/2$ with $V_j = (1-\eta)^{d_j}$, and every other line of Algorithm~\ref{alg:method} is unchanged. The mode-width identity \eqref{eq:modewidth} generalizes with $4N_jd_j^2$ replaced by the noisy Fisher information $4N_jV_j^2 d_j^2\sin^2(2d_j\theta)/(1 - V_j^2\cos^2(2d_j\theta))$, which is no longer $\theta$-independent; the grid rule \eqref{eq:gridrule} should then use its minimum over $\theta$, which is conservative and, for the visibilities of Table~\ref{tab:noise}, changes nothing.

Finally, ``exact'' refers to the optimization, not to optimality of the estimator. The likelihood \eqref{eq:lik} is the exact finite-shot sampling distribution at any shot count and $\hat\theta$ its global maximizer to a precision far below the statistical error; what is asymptotic is only the attainment of the \ac{CRLB}. At a single shot on the deepest rung no efficiency guarantee is available, which is why every comparison against the \ac{CRLB} here is empirical rather than assumed.

\section{The Ziv--Zakai bound}\label{app:zzb}

Section~\ref{sec:zzb} uses Eq.~\eqref{eq:zzb} to bound how far the estimator can be from optimal and quotes the resulting numbers. This appendix derives the bound, isolates the single inequality it rests on, explains why one computation yields every percentile, and records how $P_{\min}$ is evaluated. Throughout, $p(\cdot)$ is the prior density of the amplitude, uniform on $(a_-, a_+) = (0.1, 0.9)$, and, within this appendix only, $\varepsilon = \hat a - a$ denotes the signed estimation error of an arbitrary estimator, over the joint law of amplitude and data.

\subsection{From estimation to binary testing}\label{app:zzb:derive}

Fix a separation $g > 0$ and a location $a$, and consider deciding between
\begin{equation*}
    H_0: A = a \qquad\text{versus}\qquad H_1: A = a + g
\end{equation*}
from one run of the schedule, with the two hypotheses weighted as the prior weights them, $p(a)$ against $p(a{+}g)$. Any estimator induces a decision rule for this test---report $H_1$ when $\hat a \ge a + g/2$ and $H_0$ otherwise---whose error probability is
\begin{align}\label{eq:induced}
    P_e \;=\; \frac{1}{p(a) + p(a{+}g)}\Big\{\,& p(a)\,\Pr\big[\hat a \ge a + \tfrac{g}{2} \,\big|\, a\big] \notag\\
    &+\; p(a{+}g)\,\Pr\big[\hat a < a + \tfrac{g}{2} \,\big|\, a{+}g\big]\Big\} .
\end{align}
No test can do better than the likelihood-ratio test, so $P_e \ge P_{\min}(a, a{+}g)$. Clearing the denominator and integrating over $a$,
\begin{align}
    &\int \! p(a) \Pr\big[\varepsilon \ge \tfrac{g}{2} \big| a\big] da \;+ \int \! p(a{+}g) \Pr\big[\varepsilon < -\tfrac{g}{2} \big| a{+}g\big] da \notag \\
    &\qquad \ge\; \int \big[p(a) + p(a{+}g)\big] P_{\min}(a, a{+}g)\, da .
\end{align}
Substituting $u = a + g$ in the second integral turns the left-hand side into $\Pr[\varepsilon \ge g/2] + \Pr[\varepsilon < -g/2]$. These are disjoint events contained in $\{|\varepsilon| \ge g/2\}$, so the left-hand side is at most $\Pr[|\varepsilon| \ge g/2]$, which gives Eq.~\eqref{eq:zzb}.

Two features of this derivation make the bound useful here.

First, it contains exactly one inequality, $P_e \ge P_{\min}$, and no regularity conditions whatever: no unbiasedness, no differentiability of the likelihood in $a$, no asymptotics, no assumption that the estimate lies in any particular basin. It therefore applies to the exact global \ac{ML} estimate at single-digit shot counts, which is the regime where the \ac{CRLB} is no longer guaranteed to be attainable.

Second, the two one-sided tails must be collected separately, as above. Bounding each of them by the two-sided $\Pr[|\varepsilon| \ge g/2]$ before integrating is the natural first move and costs a factor of two in the bound---enough, at the numbers of Table~\ref{tab:zzb}, to turn a certificate into a triviality.

The looseness that remains has two sources. The rule induced by $\hat a$ in Eq.~\eqref{eq:induced} need not be a good test, so $P_{\min}$ may sit well below $P_e$; and the test is \emph{pairwise}, weighing $a$ against one rival at a time, whereas in the tail many basins compete at once. The second is the more serious, and is the caveat attached to the $99\%$ row in Sec.~\ref{sec:zzb}.

\subsection{Valley-filling and percentiles}\label{app:zzb:valley}

Write $A(g)$ for the right-hand side of Eq.~\eqref{eq:zzb}. Because $|\varepsilon| \ge g'/2$ implies $|\varepsilon| \ge g/2$ whenever $g' \ge g$, every $A(g')$ with $g' \ge g$ is also a lower bound at $g$, so the bound may be strengthened to its non-increasing envelope
\begin{equation}\label{eq:valleyfill}
    \Pr\big[|\varepsilon| \ge g/2\big] \;\ge\; \bar A(g) \;:=\; \sup_{g' \ge g} A(g') .
\end{equation}
This step is known in the Ziv--Zakai literature as \emph{valley-filling}, and it matters in general. For these schedules $A$ is non-monotone: it develops local maxima several basins out, where the deep fringes realign and the two hypotheses become hard to separate again. These are the same realignments that Eq.~\eqref{eq:chernoff} prices in the aliasing analysis. That non-monotonicity is the bound registering ambiguity (a rival several basins away is nearly as hard to reject as one inside the mode), and a local bound has no counterpart to it.

For a percentile, however, the envelope is redundant. Writing the confidence level $\delta$ as a fraction in this subsection, if $\bar A(g) \ge 1 - \delta$ then $|\varepsilon|$ exceeds $g/2$ with probability at least $1-\delta$, so the $\delta$th percentile satisfies
\begin{equation}\label{eq:zzbpct}
    \varepsilon_\delta \;\ge\; \tfrac12 \sup\big\{\, g : \bar A(g) \ge 1 - \delta \,\big\} ;
\end{equation}
but $\bar A(g) \ge t$ holds precisely when $A(g') \ge t$ for some $g' \ge g$, so $\{g : \bar A(g) \ge t\}$ and $\{g : A(g) \ge t\}$ have the same supremum. Taking the largest crossing of $A$ itself already performs the valley-filling. The distinction does matter in the more familiar mean-square form of the bound, $\mathbb{E}[\varepsilon^2] \ge \tfrac12\int_0^\infty g\,\bar A(g)\,dg$, where the envelope raises the integrand pointwise and cannot be dropped; for a threshold crossing it can. We apply the running maximum regardless, and it leaves every entry of Table~\ref{tab:zzb} unchanged.

All three rows of Table~\ref{tab:zzb} are thus one curve $A(g)$ read at $0.32$, $0.05$, and $0.01$: the bound constrains the entire survival function, and a confidence level is a choice of where to cross it.

Unlike the \ac{CRLB}, Eq.~\eqref{eq:valleyfill} bounds the survival function itself and needs no shape assumption such as the within-basin normality behind the coefficients $\kappa_\delta$ of Sec.~\ref{sec:setting}. For the same reason no Jacobian conversion arises: the bound is stated for $|\hat a - a|$ in the amplitude parametrization from the outset, whereas the \ac{CRLB} must be carried across $a = \sin\theta$ with the care described in Sec.~\ref{sec:setting}.

One property is inherited from the Bayesian formulation and should be read into the numbers: Eq.~\eqref{eq:zzb} bounds the prior-averaged error, so it bounds performance over the ensemble $a \sim \mathcal{U}(0.1, 0.9)$ rather than uniformly in $a$. The achieved constants are averages over the same ensemble, so the comparison in Table~\ref{tab:zzb} is between like quantities.

\subsection{\texorpdfstring{Evaluating $P_{\min}$}{Evaluating Pmin}}\label{app:zzb:pmin}

The normalization in Eq.~\eqref{eq:induced} cancels against the prefactor in Eq.~\eqref{eq:zzb}: writing $p_0 = p(a)$, $p_1 = p(a{+}g)$ and $P, Q$ for the two data distributions,
\begin{equation}\label{eq:errmass}
    \big[p_0 + p_1\big] P_{\min} \;=\; \sum_x \min\big\{p_0 P(x),\, p_1 Q(x)\big\} ,
\end{equation}
the unnormalized error mass. The sum is over $x = (k_1, \dots, k_L)$, so it has $\prod_j (N_j + 1)$ terms---between $10^{11}$ and $10^{22}$ for the ladders of Table~\ref{tab:scale}---and cannot be evaluated term by term. Factoring out $P$ instead,
\begin{equation}\label{eq:pminmc}
    \sum_x \min\big\{p_0 P, p_1 Q\big\} \;=\; p_0\, \mathbb{E}_{x \sim P}\Big[\min\big(1, \tfrac{p_1}{p_0} e^{\Delta(x)}\big)\Big],
\end{equation}
with the log-likelihood ratio (the symbol $\Delta$ is unrelated to the grid spacing of Appendix~\ref{app:ml})
\begin{equation}
    \Delta(x) = \sum_j \Big[ k_j \ln \frac{q_j}{p_j} + (N_j - k_j) \ln \frac{1-q_j}{1-p_j} \Big],
\end{equation}
$p_j = p_{n_j}(\arcsin a)$ and $q_j = p_{n_j}(\arcsin(a{+}g))$. The expectation is over independent binomials, which are sampled exactly, and its integrand lies in $[0,1]$; a plain Monte Carlo average therefore estimates Eq.~\eqref{eq:errmass} with no inequality anywhere and variance at most $1/4$ per sample.

The Monte Carlo evaluation is a necessity, not a refinement. The standard treatment bounds $P_{\min}$ through the Bhattacharyya coefficient $\rho = \prod_j \mathrm{BC}_j^{N_j} = e^{-E}$, the same quantity as in Eq.~\eqref{eq:chernoff}, giving
\begin{equation}\label{eq:bcbound}
    \big[p_0 + p_1\big] P_{\min} \;\ge\; \tfrac12\Big[(p_0{+}p_1) - \sqrt{(p_0{+}p_1)^2 - 4 p_0 p_1 \rho^2}\Big] ,
\end{equation}
which for $p_0 = p_1 = p$ reduces to $p\big(1 - \sqrt{1-\rho^2}\big) \approx p \rho^2/2$ at small $\rho$---the Bhattacharyya exponent charged twice, $e^{-2E}$ in place of $e^{-E}$. Against the companion upper bound $p\rho$, Eq.~\eqref{eq:bcbound} can thus fall short by as much as a factor $\rho/2 = e^{-E}/2$, and does so precisely where the exponent is large, which is most of the range of $g$. Carried through, it yields $\varepsilon_{95}\Ntot \ge 2.02$ for the flagship at $n_{\max} = 125$, below that schedule's \ac{CRLB} of $2.36$ and therefore of no use. This is the same weakness noted for the union--Chernoff estimate in Sec.~\ref{sec:ladder}, which is why we evaluate $P_{\min}$ by Monte Carlo.

\subsection{Numerical protocol}\label{app:zzb:protocol}

Separations are taken on a geometric grid $g \in [10^{-5}, 0.2]$ and amplitudes on a uniform grid over $[a_-, a_+ - g]$ with trapezoidal quadrature. Eq.~\eqref{eq:pminmc} is averaged over datasets drawn under $H_0$ at each amplitude. Then $\bar A$ is formed by a running maximum from the largest separation downward, and Eq.~\eqref{eq:zzbpct} is read off as the largest grid point of the resulting monotone curve at which the threshold is met. Table~\ref{tab:zzb} uses $520$ separations, $600$ amplitudes, and $2\times 10^3$ datasets per amplitude. Coarsening each grid by a factor of two and the Monte Carlo sample by a factor of five moves the tabulated bounds by up to $3\%$, with a seed-to-seed spread of the same order at that coarser setting. The entries are therefore converged to the order of a percent. This is small against the gaps they are used to establish but not negligible, so they should be read to that accuracy rather than to the digits shown.

\end{document}